\documentclass[letterpaper,11pt]{article}

\usepackage{ucs}
\usepackage[utf8x]{inputenc}
\usepackage{graphicx}
\usepackage{amsfonts}
\usepackage{dsfont}
\usepackage{amssymb}
\usepackage{amsmath}
\usepackage{amsthm}
\usepackage{enumerate}
\usepackage{stmaryrd}
\usepackage{fullpage}
\usepackage{ifthen}
\usepackage{subfigure}
\usepackage{epic}
\usepackage{authblk}
\usepackage{textcomp}

\usepackage{graphicx}
\usepackage{epstopdf}

\usepackage[hypertexnames=false,colorlinks=true,linkcolor=blue,citecolor=blue]{hyperref}
\usepackage[numbers,comma,square,sort&compress]{natbib}

\usepackage{color}
\usepackage{titlesec}

\graphicspath{{eps/}{pdf/}}

\newcommand{\be}{\begin{equation}}
\newcommand{\ee}{\end{equation}}

\newcommand{\floor}[1]{\lfloor #1 \rfloor}

\newcommand{\e}{\epsilon}

\newtheorem{lem}{Lemma}[section]
\newtheorem{thm}{Theorem}
\newtheorem{prop}[lem]{Proposition}

\newtheorem{rmk}[lem]{Remark}

\newtheorem{ex}[lem]{Example}

 {\begin{trivlist} \item[]{\bf Proof. }}%
{\hspace*{\fill}$\rule{.4\baselineskip}{.4\baselineskip}$\end{trivlist}}

  {\begin{trivlist}\item[]{\bf Hypothesis #1 }\em}{\end{trivlist}}

\numberwithin{equation}{section}

\title{Invasion fronts on graphs: the Fisher-KPP equation on homogeneous trees and Erd\H{o}s-R\'eyni graphs  }

\author[1]{Aaron Hoffman}
\affil[1]{\small Franklin W. Olin College of Engineering,
 Needham, MA 02492, USA}
\author[2]{Matt Holzer\footnote{mholzer@gmu.edu, 703-993-1460}}
\affil[2]{\small George Mason University, Department of Mathematical Sciences, Fairfax, VA 22030, USA}

\begin{document}
\maketitle

\begin{abstract}
We study the dynamics of the Fisher-KPP equation on the infinite homogeneous tree and Erd\H{o}s-R\'eyni random graphs.  We assume initial data that is zero everywhere except at a single node.  For the case of the homogeneous tree, the solution will either form a traveling front or converge pointwise to zero.  This dichotomy is determined by the linear spreading speed and we compute critical values of the diffusion parameter for which the spreading speed is zero and maximal and prove that the system is linearly determined.  We also study the growth of the total population in the network and identify the exponential growth rate as a function of the diffusion coefficient, $\alpha$.  Finally, we make predictions for the Fisher-KPP equation on Erd\H{o}s-R\'enyi random graphs based upon the results on the homogeneous tree. When $\alpha$ is small we observe via numerical simulations that mean arrival times are linearly related to distance from the initial node and the speed of invasion is well approximated by the linear spreading speed on the tree. Furthermore, we observe that  exponential growth rates of the total population on the random network can be bounded by growth rates on the homogeneous tree and provide an explanation for the sub-linear exponential growth rates 
that occur for small diffusion.

\end{abstract}


\section{Introduction}
We consider the Fisher-KPP equation defined on a graph and study the transition of the system from unstable to stable homogeneous states.  To provide some motivation for this study we consider the dynamics of an invasive species on a transportation network.  One can think of each node in the network as describing a physical location and the edges connecting these nodes as available transportation routes between these locales.  When the invasive species is introduced at one location in the network we expect the species to grow and spread through the network until it resides at every node in the network that is favorable to the species.  The goal would be to determine in what manner this transition occurs and to estimate the time that is required for the species to arrive at any given node in the network.  

A realistic description of this process would involve complicated network topologies and a great deal of heterogeneity -- both in the local dynamics of the species at each node as well as the transportation mechanisms available to move between modes.  Our focus in this paper is to study a simplified version of this problem where the local dynamics are identical and prescribed by logistic growth while the movement between nodes is encapsulated by diffusion.  In regards to network topologies, we concentrate mainly on the case where the network is a homogeneous tree.  This restriction allows us to make precise statements regarding the dynamics of the system, which then serve as a baseline for comparison with numerical studies of more complicated and realistic situations.

To be precise, let $G=(V,E)$ be a undirected and unweighted countable graph; we will mainly be interested in the infinite homogeneous tree, but leave the discussion general for the time being leave.  Consider a differential equation defined on the set of vertices with local dynamics prescribed by the scalar differential equation $u_t=f(u)$.  The reaction term $f(u)\in C^1([0,1])$ is assumed to be of KPP type, see \cite{fisher37,kolmogorov37}, and satisfies the conditions,
\[ \begin{array}{lc} f(0)=f(1)=0  & \\
f'(0)>0, f'(1)<0 & \\
  f(u)>0 & \text{for} \  0<u<1 \\
f(u)<f'(0)u & \text{for}\  0<u<1.
\end{array}\]
We can assume further that $f'(0)=1$ after rescaling the independent variable. In most of the numerical simulations contained in this paper we use the explicit example of logistic growth, $f(u)=u(1-u)$.  Diffusion between nodes is incorporated by the graph Laplacian, $\Delta_G=A(G)-D(G)$ where $A(G)$ is the adjacency matrix and $D(G)$ is diagonal with $i$th entry equal to the degree of the $i$th node.  Incorporating both the local dynamics and diffusion between nodes we arrive at the reaction-diffusion equation that will be the main focus of this article, 
\be u_t=\alpha\Delta_G u+f(u). \label{eq:KPPonG} \ee
The dependent variable $u_i\in\mathbb{R}$ can be thought of as representing the population of a species residing at node $i\in V$.   

Consider initial conditions where the population of the species is zero everywhere in the graph aside from one node where some positive concentration exists.  Due to diffusion, the population will spread out into the graph.


We are interested in both the {\bf pointwise} and {\bf aggregate} behavior of the population.  For the infinite tree, we will be particularly interested in the existence/ non-existence of traveling fronts that propagate asymptotically with fixed speed and replace the unstable state at zero with the stable steady state at one.  When these fronts exist, their speeds provide a measure of how long it takes for the invasive species to overtake a node as a function of the distance of that node from the node of introduction.  For sufficiently large values of the diffusion parameter $\alpha$, we find that compactly supported initial data converges pointwise to zero.  It is important to note that this is not representative of the extinction of the species.  In studying the aggregate behavior we focus on the dynamics of the total population; that is the sum of the population over all the nodes in the graph.  We find that this growth rate is always positive.  Of particular interest is how this growth rate depends on system parameters and its relation to the existence/ non-existence of traveling fronts.

The situation is somewhat different for finite connected Erd\H{o}s-R\'eyni graphs  since the steady state at one is a global attractor of the dynamics.  Here we study arrival times describing how long it takes the species to reach any given node.  We are interested in whether the arrival times are approximately linear with respect to the distances from the node of origination.  When $\alpha$ is small we observe this to be the case numerically and we draw analogies between this behavior and the propagation of traveling fronts.

%

Let us now focus on the case where $G$ is an infinite homogeneous tree where each node has degree $k+1$.  Identify one node as the root, label it $u_1$, and consider initial conditions wherein $u_1(0)=1$ with all other nodes initially equal to zero.  For this initial data, (\ref{eq:KPPonG}) can be reduced to the lattice dynamical system,
\begin{eqnarray}
  \frac{du_n}{dt}&=& \alpha\left( u_{n-1}-(k+1)u_n+ku_{n+1}\right)+f(u_n), \quad n\geq 2,\nonumber \\
  \frac{du_1}{dt}&=& \alpha(k+1)\left( -u_1+u_2\right)+f(u_1). \label{eq:tree} 
\end{eqnarray}
where $u_n(t)$ is a representative node from the set of nodes at distance $n-1$ from the root.

The evolution equation for any non-root node of the tree can be expressed as
\be \frac{du_n}{dt}= \alpha\left( u_{n-1}-2u_n+u_{n+1}\right)+\alpha(k-1)\left( u_{n+1}-u_n\right)+f(u_n).\label{eq:treefactored} \ee
In this way, the linear terms in (\ref{eq:tree}) can be viewed as a competition between diffusive and advective effects.  Take $\alpha=\frac{1}{(\Delta x)^2}$ for some small $\Delta x$.  Then  the first term on the right hand side of (\ref{eq:treefactored}) can be viewed as a discretization of the second derivative while the second term can be viewed as a discretization of the first derivative.  Therefore, in the continuum limit the dynamics are formally approximated by the solution of PDE,
\[ u_t=u_{xx}+\sqrt{\alpha}(k-1)u_x+f(u),\]
on a semi-infinite domain with no-flux boundary conditions at the left boundary.  In this scenario, the advection dominates the diffusion and localized initial data propagates to the left, i.e. up the tree, and eventually converges to zero.  On the other hand, when $\alpha$ is small and the system is near the anti-continuum limit, the reaction terms dominate the diffusive ones and localized initial data will spread down the tree.  The goal of this article is to understand the transition from spreading to pointwise convergence to zero  and, in the case of spreading,  to predict the spreading speed of the solution to (\ref{eq:tree}).

There has been a tremendous amount of effort dedicated to the study of the Fisher-KPP equation in a variety of contexts.  The system was initially studied as a PDE on the real line by Fisher \cite{fisher37} and Kolmogorov, Petrovskii and Piscunov \cite{kolmogorov37}.  In the context of lattice dynamical systems, the existence of traveling waves was established in \cite{zinner93}.  Most related to the current study is recent work in \cite{matano15} where spreading speeds for the Fisher-KPP equation on the hyperbolic space $\mathbb{H}^n$ are investigated.  Homogeneous trees are often used as a model for $\mathbb{H}^n$ and there are strong connections between the results obtained here and those in \cite{matano15}.  Indeed, in \cite{matano15} it is shown that there exists a critical diffusion coefficient above which the solution converges uniformly to zero, while for values below that threshold compactly supported initial data form a traveling front.  

The previous two decades has also witnessed an explosion of interest into the dynamics of differential equations on networks; see \cite{newman03,porter16,strogatz01,vespignani12} and the references therein for an overview of the field. Many studies of dynamical systems on networks treat each node as an individual with discrete state.  The present work deals with meta-population models, where the dependent variable at each node describes the concentration of individuals; see for example \cite{brockmann13,colizza07,hindes13} for studies of disease epidemics within the meta-population paradigm.  The work undertaken here is related to \cite{burioni12}.  There the Fisher-KPP equation on random graphs is studied with a focus on the growth rate of the overall population in the network.  Exponential growth rates of the population are studied and sub-linear growth rates are observed for small values of the diffusion parameter.  Finally, we also mention \cite{kouvaris12} where a bistable reaction-diffusion equation is considered and the existence of traveling and pinned, or stationary, fronts are exhibited on random networks and trees.

We now outline our main results, beginning with those that pertain to the case of  the infinite homogeneous tree. 

\paragraph{Linearly selected spreading speeds -- $l^1$ and $l^2$ critical diffusion coefficients}
A key feature of the dynamics of (\ref{eq:tree}) is non-monotonicity of the spreading speed.  We identify  two critical diffusion coefficients: $\alpha_1(k)$ and $\alpha_2(k)$.  The spreading speed is maximal for $\alpha_1(k)$ while for $\alpha_2(k)$ the spreading speed is zero and marks the boundary between spreading and non-spreading.  For any $0<\alpha<\alpha_2(k)$, the initial value problem (\ref{eq:tree}) forms a traveling front spreading with the linear spreading speed.  For $\alpha>\alpha_2(k)$ the solution converges pointwise exponentially fast to zero.  The decay rate of the selected front is $l^1$ critical at $\alpha_1$ and $l^2$ critical for $\alpha_2$.  

We also consider the generalization of (\ref{eq:tree}) to periodic trees, where the number of children at each level is periodic.  In this case a similar reduction to a lattice dynamical system is possible.  In analogy with the homogeneous case, we observe that the linear spreading speed for this system is zero when the selected decay rate is $l^2$ critical and the spreading speed has a critical point when  the selected decay rate is $l^1$ critical.  

It is interesting to contrast these dynamics with the case of the classical Fisher-KPP equation on the real line.  In that case, as the diffusion parameter $\alpha$ is increased the solution propagates with an increased speed.  This makes intuitive sense, as increasing the diffusion constant corresponds to increasing the mobility of the species being modeled.  On the infinite tree, it would appear that for $\alpha>\alpha_1(k)$, increasing the mobility of the species decreases the speed at which it invades.  However, we emphasize that while the solution to (\ref{eq:tree}) slows down with increasing $\alpha>\alpha_1$, the invasion of the species does not.  

When $\alpha$ is small the reactive effects dominate the diffusive ones and we obtain an asymptotic expansion of the spreading speed as $\alpha(k)\to 0$.  This expansion is independent of $k$ to leading order.  We therefore predict that dynamics of (\ref{eq:tree}) on more general networks should also be independent of the network topology for $\alpha$ sufficiently small.  We test this prediction on Erd\H{o}s-R\'enyi random graphs.

\paragraph{Growth rate of the total population}
We also study the total population on the graph -- both in aggregate and at each level of the tree.  Our main contribution here is to identify the maximal growth rate of the total population and to locate which levels of the tree this maximal population growth is concentrated on as a function of time.  

The exponential growth rate of the population in the linear system is one.  In the nonlinear system, growth is saturated at $u=1$ and we expect that the growth rate of the population may be slower in this case.  
In turns out that the critical diffusion parameter $\alpha_1(k)$
demarcates the boundary between these two possibilities.  
For $\alpha<\alpha_1(k)$, the population growth occurs primarily at the front interface and the maximal growth rate depends on the speed and decay rate of the front.  When $\alpha$ is larger than $\alpha_1(k)$, the maximal growth rate is one and occurs in a region ahead of the front interface where the components $u_n(t)$ are near zero.    The speed at which the maximal growth rate occurs can be determined in terms of the group velocity of the mode $k^n$.

\paragraph{Front propagation and population growth rates on Erd\H{o}s-R\'enyi random graphs: numerical simulations}
Finally, we address the potential relevance of our results to large, but finite, random graphs of Erd\H{o}s-R\'enyi type, \cite{erdos59}.  When the number of nodes in the network is large and the average degree of each node is small it may be reasonable to approximate the dynamics on the random graph by that of a tree with $k$ chosen appropriately.  We perform numerical simulations that suggest that for small values of $\alpha$, the dynamics of (\ref{eq:KPPonG}) are well approximated by the finite homogeneous tree.  We also consider the growth rate of the total population on the graph and demonstrate that this rate can be estimated from the growth rate on the homogeneous tree.  In analogy to the tree case, for small values of $\alpha$ the growth rate is observed to be less than the linear growth rate.

The article is organized as follows.  In Section~\ref{sec:tree}, we study the pointwise behavior of the solution on the homogeneous tree and prove Theorem~\ref{thm:spread}.  In Section~\ref{sec:pop}, we study the behavior of the total population in the tree.  In Section~\ref{sec:ER}, we apply some of the insights gained from the homogeneous tree to  Erd\H{o}s-R\'enyi random graphs.  We conclude in Section~\ref{sec:conc} with a short discussion.

\section{Front propagation  -- the linear spreading speed and proof of Theorem~\ref{thm:spread}}\label{sec:tree}
In this section, we study the dynamics of (\ref{eq:tree}) and establish the spreading speed  (or lack thereof) of solutions propagating down the tree. We first review the notion of linear spreading speeds.  We then establish some properties of the linear spreading speed for (\ref{eq:tree}) and identify the critical coefficients $\alpha_1(k)$ and $\alpha_2(k)$.  Noteworthy among these properties is the independence of the spreading speed on the constant $k$ in the asymptotic limit of $\alpha$ small.  Finally, we establish that the linear spreading speed is the speed selected in the nonlinear system and we extend the linear spreading speed to the case of periodic trees.

The Fisher-KPP equation, whether it be posed on the real line, lattice or other more exotic domain is an example of a {\em linearly determinate} system wherein the spreading speed for the full nonlinear system is equal to the spreading speed of the system linearized about the unstable state; see for example \cite{weinberger82}.  In this regard, a natural starting point for our analysis is system (\ref{eq:tree}) linearized about the unstable zero state.  Ignoring the equation for the root, we obtain the linearized problem,
\be  \frac{du_n}{dt}= \alpha\left( u_{n-1}-(k+1)u_n+ku_{n+1}\right)+u_n, \ n\geq 2.\label{eq:treelin} \ee
We seek separable solutions of the form $e^{\lambda t-\gamma n}$, $\gamma>0$, and obtain a solvability condition that is known as the dispersion relation, 
\be d(\lambda,\gamma)= \alpha\left(e^{\gamma}-k-1+ke^{-\gamma}\right)+1-\lambda.\label{eq:ulin} \ee
Simple roots of the dispersion relation relate exponentially decaying modes $e^{-\gamma n}$ to their temporal growth rate $e^{\lambda t}$.  Double roots, or more precisely pinched double roots, of (\ref{eq:ulin}) identify pointwise growth rates of compactly supported initial data; see for example \cite{briggs,bers84,brevdo96,sandstede00,vansaarloos03,holzer14}.  Double roots differentiate between absolute and convective instabilities.  In the context of (\ref{eq:treelin}), an absolute instablity refers to exponential growth of both the solution in norm ($l^2$ for example) as well as pointwise at each fixed node.  On the other hand, a convective instablity grows in norm, but is transported away from its original location so that the solution decays pointwise at each node.   A convective instability is one where the dispersion relation has simple roots with positive temporal growth rates, but all pinched double roots have negative growth rates, i.e. $\mathrm{Re}(\lambda)<0$.

The linear spreading speed is the speed at which compactly supported initial data spreads in the asymptotic limit as $t\to\infty$.  
This speed can be thought of as the speed for which the system transitions from absolute to convective instability and is defined by the presense of a pinched double root with purely imaginary, or in our case zero,  temporal growth rate.  We therefore consider the dispersion relation (\ref{eq:ulin}) in a co-moving frame. Let $d_{s}(\lambda,\gamma)=d(\lambda-s\gamma,\gamma)$ denote the dispersion relation in this frame.  Then the linear spreading speed for (\ref{eq:ulin}) is the speed $s_{lin}$ for which a $\gamma_{lin}$ simultaneously satisfies the equations,
\[ d_{s_{lin}}(0,\gamma_{lin})=0, \quad \partial_\gamma d_{s_{lin}}(0,\gamma_{lin})=0,\]
together with a pinching condition.  Pinching here refers to the requirement that the roots participating in the double root can be traced to opposite sides of the imaginary axis as $\mathrm{Re}(\lambda)\to \infty$.  

In order to introduce some further terminology we also mention an alternate, and equivalent, formulation of the linear spreading speed  as a minimizer of the {\em envelope velocity},
\[ s_{env}(\gamma)=\alpha\left( \frac{e^{\gamma}-k-1+ke^{-\gamma}}\gamma\right)+\frac{1}{\gamma}.\]
The envelope velocity prescribes the speed at which an exponential solution $e^{-\gamma n}$ propagates in the linear equation (\ref{eq:treelin}).  For systems that obey the comparison principle, any positive solution of the linearized system places an upper bound on the speed of propagation in the nonlinear system and so minimizing over all such speeds provides an upper bound.  In this way, the linear spreading speed can be alternatively defined as 
\[ s_{lin}=\min_{\gamma\in\mathbb{R}^+} s_{env}(\gamma).\]
The equivalence of these two formulations can be observed by imposing that 
\[0= \partial_\gamma s_{env}(\gamma)= \frac{1}{\gamma}\left(\frac{d\lambda}{d\gamma}-s_{env}\right).\]
The quantity $\frac{d\lambda}{d\gamma}$ is known as the {\em group velocity} and so the linear spreading speed occurs when the group velocity equals the envelope velocity.  In a frame moving at the envelope velocity, this then implies zero group velocity and hence that $\gamma$ is a pinched double root with $s=s_{env}$.  The notion of group velocity will arise again in Section~\ref{sec:pop}.

Moving forward, we will use the formulation in terms of pinched double roots and define 
\be F(s,\gamma,\alpha)=\left(\begin{array}{cc} \alpha\left(e^{\gamma}-k-1+ke^{-\gamma}\right)-s\gamma+1 \\ \alpha\left(e^{\gamma}-ke^{-\gamma}\right)-s \end{array}\right).\label{eq:F} \ee

\subsection{Properties of the linear spreading speed for (\ref{eq:treelin})}

Inspection of (\ref{eq:F}) reveals that explicit expressions for the linear spreading speed  are not available.  We compute solutions numerically and plot the linear spreading speed in Figure~\ref{fig:linSS}.  Two features of the linear spreading speed are readily apparent: it is non-monotone as a function of $\alpha$ and becomes negative above some critical value.   In what follows, we identify the linear spreading speed at two critical values of the diffusion parameter: the diffusion coefficient leading to the fastest spreading speed we denote $\alpha_1(k)$ whereas the diffusion coefficient with zero spreading speed we denote $\alpha_2(k)$.  The notation is chosen since we will show that  for $\alpha_1$ and $\alpha_2$  the decay rates of the fronts are critical in $l^1$ and $l^2$ respectively. Recall the definition of the $l^p$ norm on $G$,
\[ ||u||_p=\left( \sum_{v\in V}  |u_v|^p \right)^{\frac{1}{p}}.\]
Under our assumption that the values of $u$ at all nodes on each level of the tree are equivalent, the $l^p$ norm is expressed as
\be ||u||_p=\left( \sum_{n=1}^\infty k^{n-1}|u_n|^p\right)^{\frac{1}{p}}. \label{eq:lp} \ee

\begin{figure}[ht]
\centering
   \includegraphics[width=0.6\textwidth]{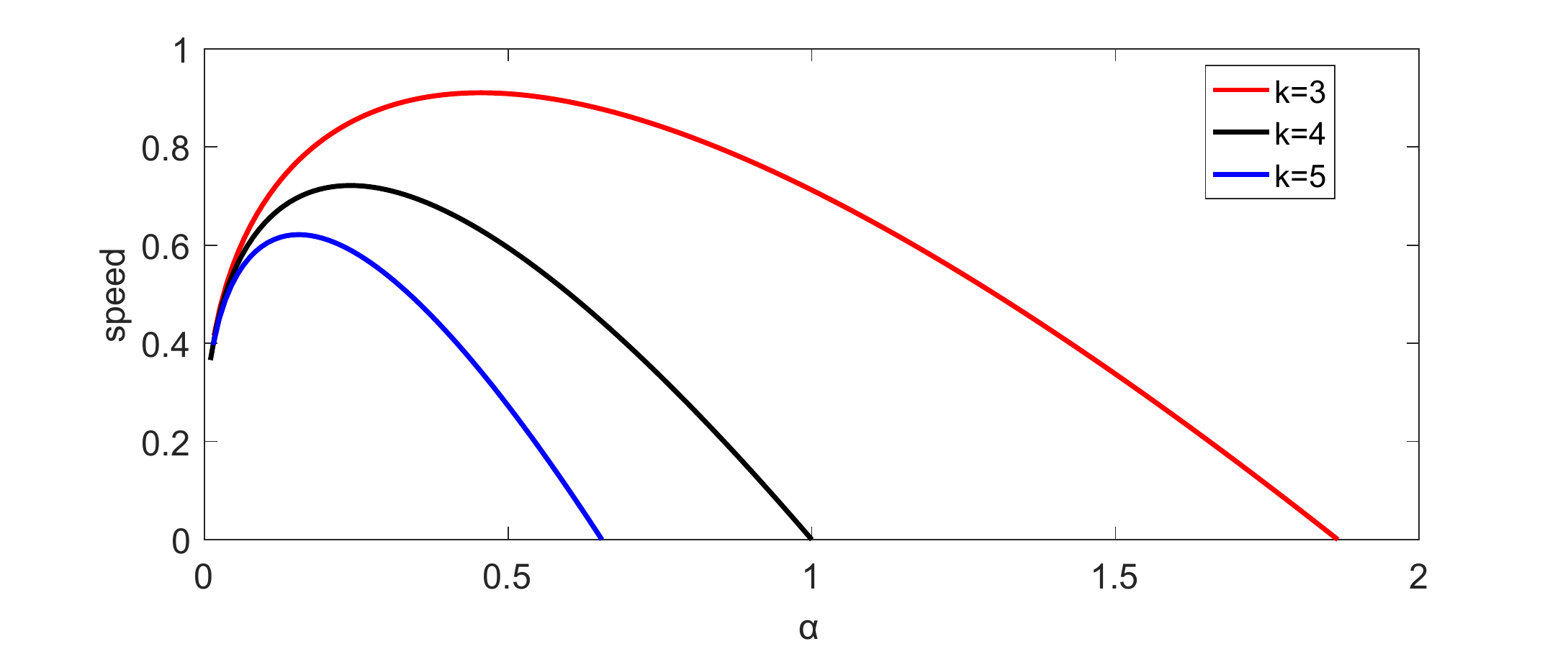}
\caption{ The linear spreading speed for (\ref{eq:ulin}), calculated numerically as a function of $\alpha$ for $k=3$ (red), $k=4$ (black) and $k=5$ (blue).  Note the critical values $\alpha_2(k)$ for which the spreading speed is zero and $\alpha_1(k)$ where the speed is maximal.  Also note that as $\alpha\to 0$, these spreading speeds appear to approach a common curve. 
}
\label{fig:linSS}
\end{figure}

\paragraph{The critical coefficient $\alpha_2(k)$}
We define $\alpha_2$ as the critical diffusion parameter for which there exists a pinched double root on the imaginary axis for $s=0$.  To compute $\alpha_2(k)$, let $s=0$ and consider the second component of $F$ set equal to zero, 
\[ \alpha\left(e^{\gamma}-ke^{-\gamma}\right)=0,\]
from which we compute the critical decay rate
\[ \gamma_2=\frac{1}{2}\log(k).\]
Plugging this into the dispersion relation, again with $s=0$,  we find 
\be \alpha_2=\frac{1}{k+1-2\sqrt{k}}.\ee
Note that $\gamma_2$ is the $l^2$-critical exponential decay rate, in the sense that any steeper exponential is $l^2$ on the tree while for any weaker exponential is not, see (\ref{eq:lp}). 
\paragraph{The critical coefficient $\alpha_1(k)$}
We again seek pinched double roots. Observe that for  $\gamma_1=\log(k)$,
\[ F(s,\log(k),\alpha)=\left(\begin{array}{cc} -s\log(k)+1 \\ \alpha\left(k-1\right)-s \end{array}\right),\]
then we obtain a solution $F(s_1,\log(k),\alpha_1)=0$, for
\[ s_1=\frac{1}{\log(k)}, \quad \alpha_1=\frac{1}{(k-1)\log(k)}.\]
Compute the Jacobians
\[ D_{s,\gamma}F(s_1,\log(k),\alpha_1)=\left(\begin{array}{cc} -\log(k) & 0 \\ -1 & \alpha_1(k+1)\end{array}\right), \quad D_{\alpha}F(s_1,\log(k),\alpha_1)=\left(\begin{array}{c} 0 \\ k-1 \end{array}\right),\]
from which the Implicit Function Theorem implies that double roots can be continued as functions $s(\alpha)$ and $\gamma(\alpha)$ with
\[ s'\left(\alpha_1\right)=0.\]
Note that $\gamma_1$ is the $l^1$-critical decay rate, recall (\ref{eq:lp}).

\begin{lem} The linear spreading speed is monotone increasing for $\alpha<\alpha_1(k)$ and is monotone decreasing for $\alpha_1(k)<\alpha<\alpha_2(k)$.  
\end{lem}
\begin{proof} Solving the second equation in $F=0$ for $s$ and plugging into the first equation we obtain an implicit equation for $\gamma$ and $\alpha$ given by 
\[ \alpha \left( (\gamma-1)e^{\gamma}+(k+1)-k(\gamma+1)e^{-\gamma}\right)=1.\]
Fix $k$ and define
\[ G(\gamma)=(\gamma-1)e^{\gamma}+(k+1)-k(\gamma+1)e^{-\gamma}.\]
Note that $G(0)=0$, $\lim_{\gamma\to\infty} G(\gamma)=\infty$ and that $G'(\gamma)>0$.  Therefore $G$ is monotone and  invertible on its range.  Consequently,  for every $\gamma>0$ there exists a $\alpha$ for which  $F\left(\alpha(e^{\gamma}-ke^{-\gamma}),\gamma,\alpha\right)=0$.  We observe that the selected decay rate $\gamma$ is a decreasing function of $\alpha$ and that the spreading speed $\alpha(e^\gamma-ke^{-\gamma})$ may be negative.  

Applying the Implicit Function Theorem we find an expression for 
\[ \frac{ds_{lin}}{d\alpha}=\frac{1}{\gamma\alpha}\left(e^{\gamma}-k-1+ke^{-\gamma}\right).\]
We observe that this zero exactly for $\alpha=\alpha_1$.  Differentiating the expression in the parenthesis we obtain
\[ e^\gamma-ke^{-\gamma}.\]
This derivative changes sign exactly at $\alpha=\alpha_2$ and thus $e^{\gamma}-k-1+ke^{-\gamma}$ is monotone decreasing for $\alpha<\alpha_2$ and the result follows.   
\end{proof}

\paragraph{Asymptotics for small $\alpha$}
Consider $0<\alpha\ll 1$.  We will derive a the leading order term in an asymptotic expansion for the linear spreading speed in this limit.  Consider the implicit equation
\[ \alpha G(\gamma)=1.\]
In the limit as $\alpha$ tends to zero, it must be the case that $G(\gamma)\to\infty$.  Since $\gamma>0$,  $G(\gamma)$ is dominated by the term $\gamma e^\gamma$ and we expect
\[ \gamma e^{\gamma}\approx \frac{1}{\alpha}.\]
Inverting, the selected decay rate is to leading order
\[ \gamma=W\left(\frac{1}{\alpha}\right),\]
where $W$ is the Lambert W-function.  Note that $s=\alpha e^\gamma$ to leading order and obtain
\[ s=\frac{1}{W\left(\frac{1}{\alpha}\right)},\]
to leading order. We remark that the linear spreading speed is independent of the degree $k$ to leading order when $\alpha\ll 1$.  This means that for $\alpha$ small the spreading speed of the system is close to the spreading speed of the Fisher-KPP equation on a lattice in the anti-continuum limit and the arrival time to any node in the tree is dominated by its distance from the root.  This has important implications for applications to more general networks, which we will return to later.

\paragraph{Linear Determinacy}

We have thus far focused entirely on the spreading properties of the linearized system.  We now turn our attention to the nonlinear system and introduce some standard terminology.  Define the invasion point for (\ref{eq:tree}) as
\[ \kappa(t)=\sup_{n\in\mathbb{N}}\left\{ n \ | \ u_n(t)>\frac{1}{2} \right\}.\]
The {\em selected spreading speed} is then defined as
\[ s_{sel}=\lim_{t\to\infty} \frac{\kappa(t)}{t}.\]

We prove the following theorem which establishes that the selected spreading speed is the linear spreading speed.
\begin{thm}\label{thm:spread} Consider (\ref{eq:tree}) with initial condition $u_1(0)=1$ and $u_j(0)=0$ for all $j\neq 1$.  Suppose that $0<\alpha<\alpha_2(k)$. Then the selected spreading speed equals the linear spreading speed, i.e.  $s_{sel}=s_{lin}$.  For $\alpha>\alpha_2(k)$, we have that $u_n(t)\to 0$ uniformly as $t\to\infty$.  
\end{thm}

Theorem~\ref{thm:spread} states that the system is {\em linearly determinate} and that the spreading speed of the nonlinear system is the same as the system linearized about the unstable state.  We remark that had (\ref{eq:tree}) been posed on $\mathbb{Z}$, rather than $\mathbb{N}$,  then this result would be a direct consequence of \cite{weinberger82}.  Rather than adapt that proof to this context, we will construct explicit sub and super solutions that bound the spreading speed of system (\ref{eq:tree}).  Since the proof proceeds along familiar lines, we delay its presentation until Appendix~\ref{sec:proof}.

\subsection{The periodic tree }
We now consider the generalization to an infinite rooted tree where the number of nodes per level is periodic.  We again assume initial data consisting of one at the root and zero elsewhere and obtain a similar reduction to a lattice dynamical system with the evolution given by
\begin{eqnarray}
  \frac{du_n}{dt}&=& \alpha\left( u_{n-1}-(k_n+1)u_n+k_nu_{n+1}\right)+f(u_n), \quad n\geq 2,\nonumber \\
  \frac{du_1}{dt}&=& \alpha(k_1+1)\left( -u_1+u_2\right)+f(u_1). \label{eq:pertree} 
\end{eqnarray}
where $k_n$ is periodic in $n$ with periodic $m$, i.e. $k_{n+m}=k_n$.  Our main goal is to study the linear spreading speed as a function of $\alpha$ and to show that the $l^2$ critical front is stationary while the $l^1$ front is a critical point of the linear spreading speed.  

To this end, we linearize (\ref{eq:pertree}) about the unstable state at zero and obtain the following equation at any (non-root) node
\[ \frac{du_n}{dt}= \alpha\left( u_{n-1}-(k_n+1)u_n+k_nu_{n+1}\right)+u_n.  \]
We seek separable solutions of the form 
\be u_n(t)=e^{-\gamma (n-st)}U_n,\label{eq:sov} \ee
where $\gamma>0$ and $U_n$ is periodic with period $m$. Since $U_n$ is periodic we focus on one segment of length $m$ which we denote $U=(U_1,U_2,\dots,U_{m})^T$.  Plugging in the ansatz (\ref{eq:sov}), $U$ must satisfy
\[ \gamma s U = \alpha M U+U,\]
where for $m\geq 3$, 
\[ M_{ij}(\gamma) = \left\{ \begin{array}{ll} e^{\gamma} & j = i-1 \mod m \\ -(k_i + 1) & j = i \\ k_i e^{-\gamma} & j = i+1 \mod m \\ 0 & \mbox{ else} \end{array}. \right. \]
Note that $M$, after adding an appropriate multiple of the identity is positive, so it has a real principle eigenvalue $\lambda(\gamma,p)$ and corresponding  principle eigenvector with positive entries, which we denote $U$.  In this manner, we obtain a scalar dispersion relation that relates the decay rate $\gamma$ to the envelope speed $s$, 
\[ d(\gamma,s)=\alpha \lambda(\gamma)-s\gamma+1=0.\]
Pinched double roots are therefore solutions of the system of equations 
\be F(s,\gamma,\alpha)=\left(\begin{array}{cc} \alpha\lambda(\gamma)-s\gamma+1 \\ \alpha\partial_\gamma \lambda(\gamma)-s \end{array}\right).\label{eq:Ftree} \ee
We are interested in the linear spreading speeds corresponding to $l^p$ critical decay rates.  Critical fronts in $l^p$ have decay rate 
\[ \gamma_p=\frac{\sum_{i=1}^{m} \log(k_i)}{mp}.\]
Then $\alpha_p$ and $s_p$ can be defined through the relationship $F(s_p,\gamma_p,\alpha_p)=0$ and we may express $F$ entirely as a function of $p$.  Set the second component of $F$ equal to zero, from which 
\[ s_p=-\alpha_p m\frac{p^2  \lambda'(p)}{\sum_{i=1}^{m} \log(k_i)}.\]
Then for the first component of $F$ to be zero it is required that 
\[ \alpha_p=\frac{-1}{\lambda(p)+\lambda'(p)p}.\]
Combining these two expressions we obtain an equation for the linear spreading speed as a function of $p$, in terms of the principle eigenvalue of $M$,  
\[ s_p= \frac{m}{\sum_{i=1}^{m} \log(k_i)}\frac{p^2  \lambda'(p)}{\lambda(p)+\lambda'(p)p}.\]
We now express the matrix $M$  in terms of $p$. 
We first note that 
\[ M_{ij}(\gamma_p) = \left\{ \begin{array}{ll} K^{\frac{1}{p}} & j = i-1 \mod m \\ -(k_i + 1) & j = i \\ k_i K^{-\frac{1}{p}} & j = i+1 \mod m \\ 0 & \mbox{ else} \end{array}, \right. \]
where $K=\left(\Pi_{i=1}^{m} k_i\right)^{\frac{1}{m}}$. 
Before stating our results, we show that $M(\gamma_p)$ is similar to $M(\gamma_q)$, where $q$ is the H\"{o}lder dual of $p$.  To do this, we compute the determinant, 
\[ \mathrm{det}(M(\gamma))=(-1)^{m+1}\left( k_1\dots k_m e^{-m\gamma}+e^{m\gamma}\right)+\Delta(k_1,\dots,k_m), \]
where the terms $\Delta(k_1,\dots,k_m)$ do not depend on $\gamma$. 
With $\gamma=\gamma_p$, we find 
\[  \mathrm{det}(M(\gamma_p)) =(-1)^{m+1}\left( (k_1\dots k_m)^{1-\frac{1}{p}} + (k_1\dots k_m)^{\frac{1}{p}} \right)+\Delta(k_1,\dots,k_m).\]
Consequently, the characteristic polynomial is invariant when $p$ is replaced with its H\"{o}lder dual $q$ and therefore the matrices $M(\gamma_p)$ and $M(\gamma_q)$ are similar.

\begin{prop}
$s_2 = 0$.
\end{prop}
\begin{proof}
Let $q$ be the H\"{o}lder dual of $p$.  Making use of the fact that $M(p)$ and $M(q)$ are similar it follows that $\lambda(p) = \lambda(q)$ and hence that $\lambda'(p) = \lambda'(q) q'$.  Since $q' = -(q/p)^2$ this rewrites as 
\[ p^2\lambda'(p) + q^2\lambda'(q) = 0 \]
In the case $p = q = 2$ this reduces to $\lambda'(2) = 0$.  Thus $s_2 = 0$ as desired.
\end{proof}

\begin{prop}
$s_1 = \frac{m}{\sum \log k_i}$.  Further $s_1' = 0$.
\end{prop}
\begin{proof}
First, we compute
\[ s_p'=\frac{m}{\sum_{i=1}^{m-1} \log(k_i)} \frac{p\lambda(2\lambda'+p\lambda'')}{(p\lambda' + \lambda)^2} \]
In the case $p = \infty$ ($1$) the matrix $M$ has row (column) sum zero, hence the vector of all ones is in its (left) null space.  As this vector is positive, it is necessarily the principal (left) eigenvector, hence zero is the principal eigenvalue.  With $\lambda = 0$ the expression for $s_p$ reduces to $mp\frac{1}{\sum \log k_l}$.  Further $s_p'$ vanishes when $\lambda = 0$.

\end{proof}
We now present some results for periodic trees with period two, where explicit expressions can be obtained.
\begin{ex}\label{ex:per} Consider the periodic tree with $m=2$, where the degrees alternate between $k_1$ and $k_2$.  We compute
\[ \lambda(\gamma)=-\frac{k_1+k_2+2}{2}+\sqrt{e^{2\gamma}+\frac{1}{4}(k_1-k_2)^2+k_1+k_2+k_1k_2e^{-2\gamma}}\]
An explicit calculation reveals
\[ \gamma_2=\frac{1}{4}\log(k_1k_2),\quad \alpha_2=\frac{1}{\frac{k_1+k_2+2}{2}-\sqrt{\frac{1}{4}(k_1-k_2)^2+k_1+k_2+2\sqrt{k_1k_2}}}, \]
and
\[ \gamma_1=\frac{1}{2}\log(k_1k_2),\quad \alpha_1=\frac{k_1+k_2+2}{(k_1k_2-1)\log{k_1k_2}} \]
We plot $\alpha_1$ and $\alpha_2$ in Figure~\ref{fig:per}.
\begin{figure}
   \includegraphics[width=0.45\textwidth]{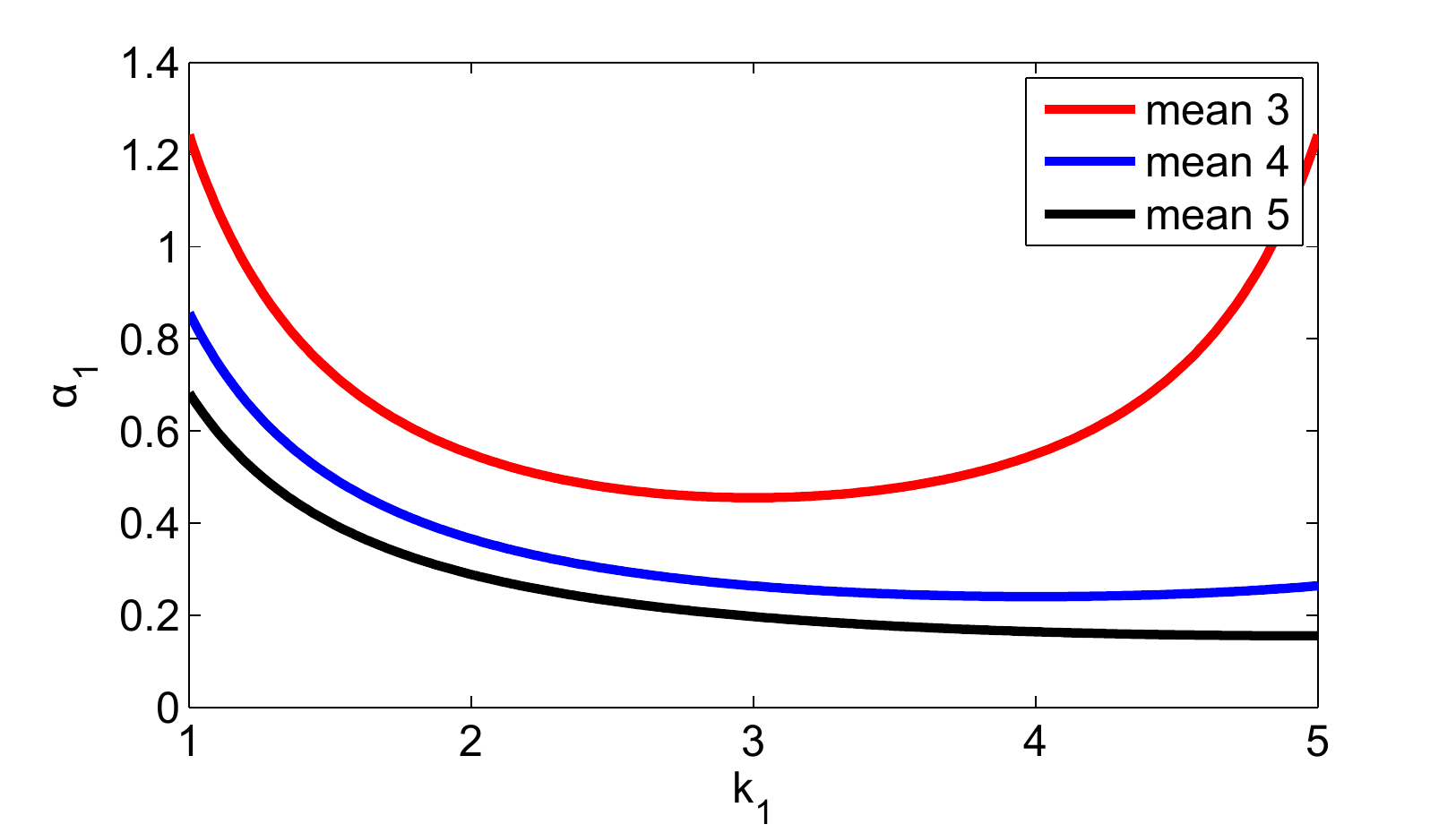} \hfill
   \includegraphics[width=0.45\textwidth]{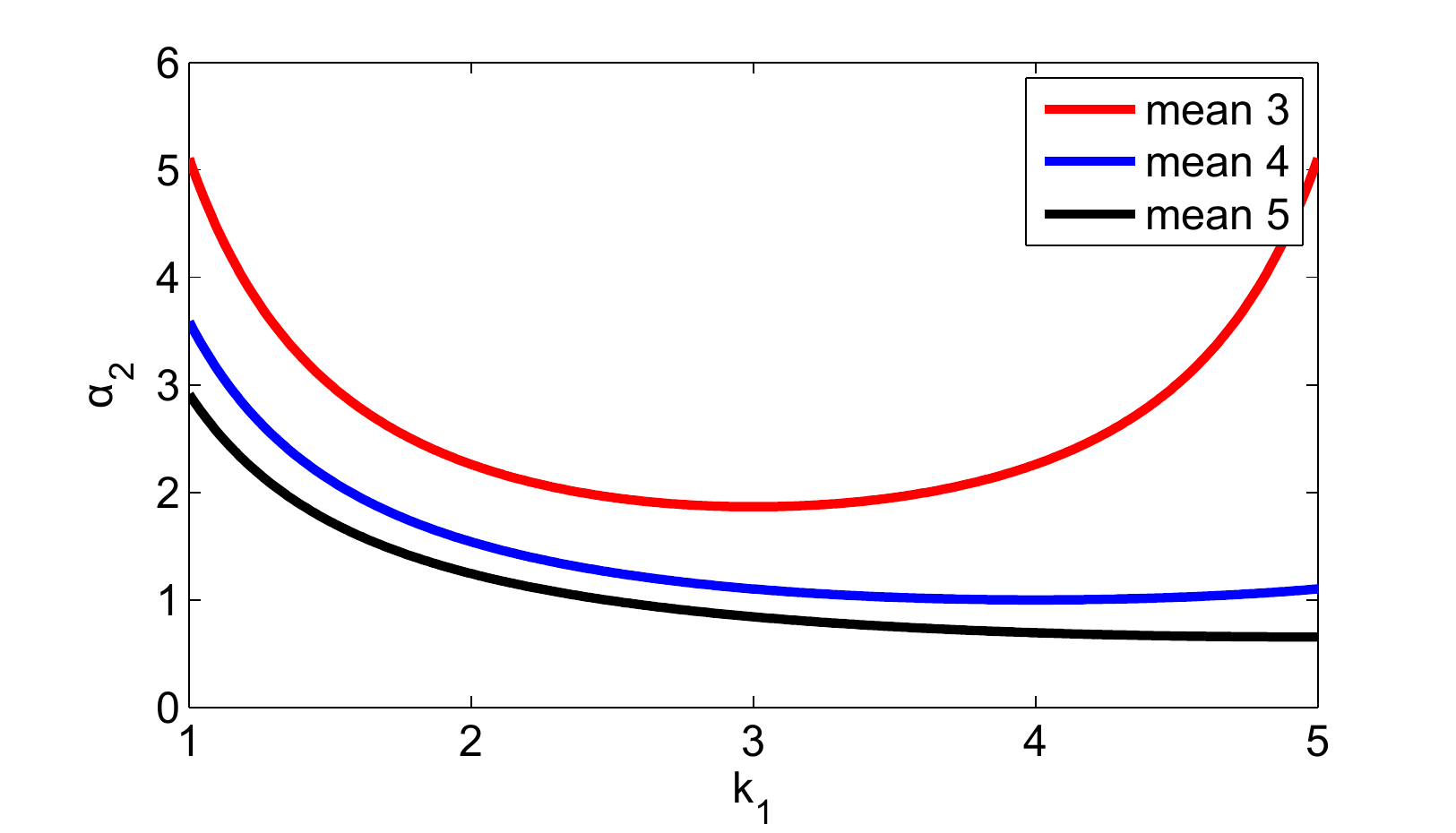}
\caption{Critical rates of diffusion for period trees with period $m=2$.  On the left, we plot $\alpha_1$ as a function of $k_1$ with   $k_2$ fixed to preserve the mean degree.  On the right, we plot $\alpha_2$ as a function of $k_1$.  Note that in both case the periodic heterogeneity increases the critical diffusion rates.  
}
\label{fig:per}
\end{figure}

Note that in both cases, the heterogeneity induced by the periodic coefficients leads to increases in the critical values as compared to the case of the homogeneous tree.  Take for example $k_1=1$, $k_2=5$ and compare $\alpha_2$ to the critical value for the homogeneous tree with degree $k=3$.  In the homogeneous case, the transition occurs at approximately $1.86$ whereas the periodic case has a critical value near $5.1$.  Heterogeneity in the tree can therefore allow for front propagation over a much larger region in parameter space.  
\end{ex}

\section{Aggregate dynamics: growth of the total population}\label{sec:pop}
In Section~\ref{sec:tree}, we determined the pointwise behavior of $u_n(t)$.  We now turn our attention to the growth rate of the total population within the homogeneous tree.  Let
\[ P(t)=\sum_{n=1}^\infty k^{n-1} u_n(t).\]
Note that $P(t)$ is equivalent to the $l^1$ norm of the solution.  We will find it convenient to express $P(t)$ in terms of 
 $\sum_{n=1}^\infty w_n(t)$, where $w_n(t)$ is the total population at level $n$ and is related to $u_n$ via  
\be w_n(t)=k^{n-1} u_n(t).\label{eq:utow} \ee
The evolution of $w_n(t)$ is given by the  lattice dynamical system
\begin{eqnarray}
  \frac{dw_n}{dt}&=& \alpha\left( kw_{n-1}-(k+1)w_n+w_{n+1}\right)+k^nf(k^{-n}w_n), \quad n\geq 2,\nonumber \\
  \frac{dw_1}{dt}&=& \alpha\left( -kw_1+w_2\right)+kf(k^{-1}w_1). \label{eq:wtree} 
\end{eqnarray}
This system is equivalent to (\ref{eq:tree}) after transforming to an exponentially weighted space.  

In Figure~\ref{fig:spacetime}, we show numerical simulations of (\ref{eq:tree}) and (\ref{eq:wtree}), with $w_n$ normalized.  We observe that for small values of $\alpha$, the maximal growth rate appears to occur at the front interface, while for larger values of $\alpha$ the maximal growth rate is achieved ahead of the front interface.  This transition occurs at $\alpha=\alpha_1(k)$.  We have two primary goals in this section.  These are to identify the maximal growth rate of $P(t)$ and determine the location in space-time for which the maximal growth rate is achieved.  As was the case in Section~\ref{sec:tree}, we begin with the system linearized about the unstable state and then consider the implications for the nonlinear system.  

\subsection{Linear growth rates  for $P(t)$}
The growth rate of $P(t)$ for the linear problem 
\[ u_t=\alpha \Delta_Gu+u,\]
can be determined from the $l^1$ spectrum of the linear operator $\alpha \Delta_G+\mathrm{Id}$.  The spectrum of the graph Laplacian on homogeneous trees has been the focus of intense study and its spectrum for any $l^p$ space is known explicitly, see for example \cite{mohar89}.  The $l^1$ spectrum of $\Delta_G$ is an ellipse
\[ \sigma_1(\Delta_G)=\left\{ z\in\mathbb{C} \ | \ |z-k-1-2\sqrt{k}|+|z-k-1+2\sqrt{k}|\leq 2(k+1)\right\}, \]
which is contained in the left half of the complex plane aside from the point at zero.  As a result, the $l^1$ spectrum of $\alpha \Delta_G+\mathrm{Id}$ extends to $\lambda=1$ and we expect a maximal exponential growth rate of $e^t$.

Of course, we are ultimately interested in the population of the nonlinear system. It is conceivable that the nonlinearity could constrain the growth in such a way that slower growth rates for $P(t)$ are observed.  To determine whether this occurs, we require pointwise information on the growth rates of the total population.  

\subsection{Pointwise growth rates and spreading speeds for the total population}

In a fashion analogous to the system for $u_n$, we can linearize (\ref{eq:wtree}) about the unstable zero state and derive a dispersion relation for  $w_n$,
\be \tilde{d}(\lambda,\gamma)=\alpha\left(ke^{\gamma}-(k+1)+e^{-\gamma}\right)+1-\lambda,\ee
as well as the equivalent dispersion relation in a co-moving frame $\tilde{d}_s(\lambda,\gamma)=\tilde{d}(\lambda+s\gamma,\gamma)$.
We again remark that transforming from $u_n$ to $w_n$ amounts to transforming $u_n$ to a weighted space.  Any pinched double root $(\lambda^*,\gamma^*)$ of the dispersion relation $d_s(\lambda,\gamma)$ corresponds to a pinched double root of the dispersion relation $\tilde{d}_s$ for $(\lambda,\gamma)=(\lambda^*+s\log{k},\gamma^*-\log{k})$; see the following calculation

\begin{eqnarray*} \left(\begin{array}{c} d_s(\lambda^*,\gamma^*) \\
\partial_\gamma d_s(\lambda^*,\gamma^*)\end{array}\right)&=& 
\left(\begin{array}{c} 
\alpha\left(e^{\gamma^*}-k-1+ke^{-\gamma^*}\right)+1-s\gamma^*-\lambda^* \\
\alpha\left(e^{\gamma^*}-ke^{-\gamma^*}\right)-s \end{array}\right) \\
& =& \left(\begin{array}{c} 
\alpha\left(ke^{\gamma^*-\log{k}}-k-1+e^{-\gamma^*+\log{k}}\right)+1-s(\gamma^*-\log{k})-s\log{k}-\lambda^* \\
\alpha\left(ke^{\gamma^*-\log{k}}-e^{-\gamma^*+\log{k}}\right)-s \end{array}\right) \\
&=& \left(\begin{array}{c} \tilde{d}_s(\lambda^*+s\log{k},\gamma^*-\log{k})  \\
\partial_\gamma \tilde{d}_s(\lambda^*+s\log{k},\gamma^*-\log{k}) \end{array}\right)
\end{eqnarray*}
Recall that the $\lambda^*$ value gives the pointwise growth rate of the solution in a frame moving with speed $s$.  The decay (growth) rate of any pinched double root for the local population $u_n(t)$ is thus increased by $s\log(k)$ in the weighted space.  
Thus,  pointwise growth rates of $w_n$ may be positive even when growth rates for $u_n$ are negative.  We have the following.  

\begin{lem}\label{lem:sg} The linear growth rate is maximized in a frame moving at speed 
\[ s^*=\alpha(k-1),\]
i.e. at the $u$-group velocity of the mode $e^{-n\log{k}}$. 
\end{lem}
\begin{proof} Note that $\gamma=0$ and $\lambda=1$ is always a root of the dispersion relation, i.e. $\tilde{d}(1,0)=0$.  In the original frame of reference, this does not correspond to a pinched double root since $\partial_\gamma \tilde{d}(1,0)=\alpha(1-k)\neq 0$.  However, passing to a frame moving with speed $s^*$ we find $\tilde{d}_{s^*}(1,0)=0$ and $\partial_\gamma \tilde{d}_{s^*}(1,0)=0$.  
\end{proof}

\begin{rmk} Lemma~\ref{lem:sg} states that the maximal growth rate of a linear problem in an exponentially weighted space occurs on a ray in space-time moving with the group velocity of the exponential weight.  This is a general fact and there are several alternative avenues by which to verify the statement of Lemma~\ref{lem:sg}.  Similar results are often obtained by a asymptotic analysis of the Green's function via stationary phase.  In the context of a homogeneous tree, this fact can be verified from a direct analysis of the heat kernel, see for example \cite{chung99,chinta15}. For the case of parabolic partial differential equations, this was shown in Lemma 6.3 of \cite{holzer14}.  Finally, we mention \cite{holzer16}, where the appearance of pointwise instabilities when transforming to weighted spaces leads to some unexpected spreading speeds in a system of coupled Fisher-KPP equations.  
\end{rmk}

We now compare the speed $s^*$ to the linear spreading speed $s_{lin}$ and make the the following observations,
\begin{itemize}
\item For $\alpha=\alpha_1$, the maximal growth rate is realized at the linear spreading speed; that is, $s^*=s_1$
\item For $\alpha<\alpha_1$ we have that $s^*<s_{lin}$, while $s^*>s_{lin}$ for $\alpha>\alpha_1$. 
\end{itemize}

These observations lead us to conjecture that the the population grows fastest when $s^*>s_{lin}$ since in this regime the maximal growth rate is achieved ahead of the front interface where the nonlinearity is negligible.

\subsection{Growth rates in the nonlinear problem}\label{sec:GRnonlinear}
We return our attention to the full nonlinear system and study the growth rate of the total population.  We expect that the growth rates predicted by the linearized system to place an upper bound on the growth rates in the nonlinear system.  

A key distinction is made between $\alpha<\alpha_1(k)$ and $\alpha>\alpha_1(k)$.  For $\alpha$  below the $l^1$ critical value, the maximal linear growth rate at level $n$ is achieved in a frame of reference that propagates {\em slower} than the linear spreading speed.  In view of Theorem~\ref{thm:spread}, this means that the invasion speed of the species exceeds $s^*$.  Furthermore, since $\gamma_{lin}>\log k$ we expect the mass in the system to be added at the front interface.  The dynamics of the total population can then be approximated by
\[ P(t)\approx \sum_{n=1}^\infty k^{n-1}\min\{1,e^{-\gamma_{lin}(n-s_{lin}t)}\}.\]
When $\gamma_{lin}>\log(k)$ then this sum converges and we expect
\[ P(t)\approx \sum_{n=1}^{\floor{ s_{lin}t}} k^{n-1}\approx\frac{k^{s_{lin}t}-1}{k-1}.\]
Therefore, we expect a maximal growth rate of $e^{\log(k)s_{lin} t}$ when $\gamma_{lin}>\log(k)$.  
Note that for $\alpha_1$, it is true that $\log(k)s_1=1$.  

For $\alpha>\alpha_1(k)$, the speed $s^*$ exceeds that of $s_{lin}$ and therefore the maximal linear growth rate is achieved ahead of the front interface.  Since the solution $u_n(t)$ is pointwise close to zero in this region, we expect that the linear dynamics remain valid here and that $P(t)\approx Ce^t$ for these parameter values.  This is confirmed in numerical simulations, see Figure~\ref{fig:speedmass}.

\begin{figure}
\includegraphics[width=0.3\textwidth]{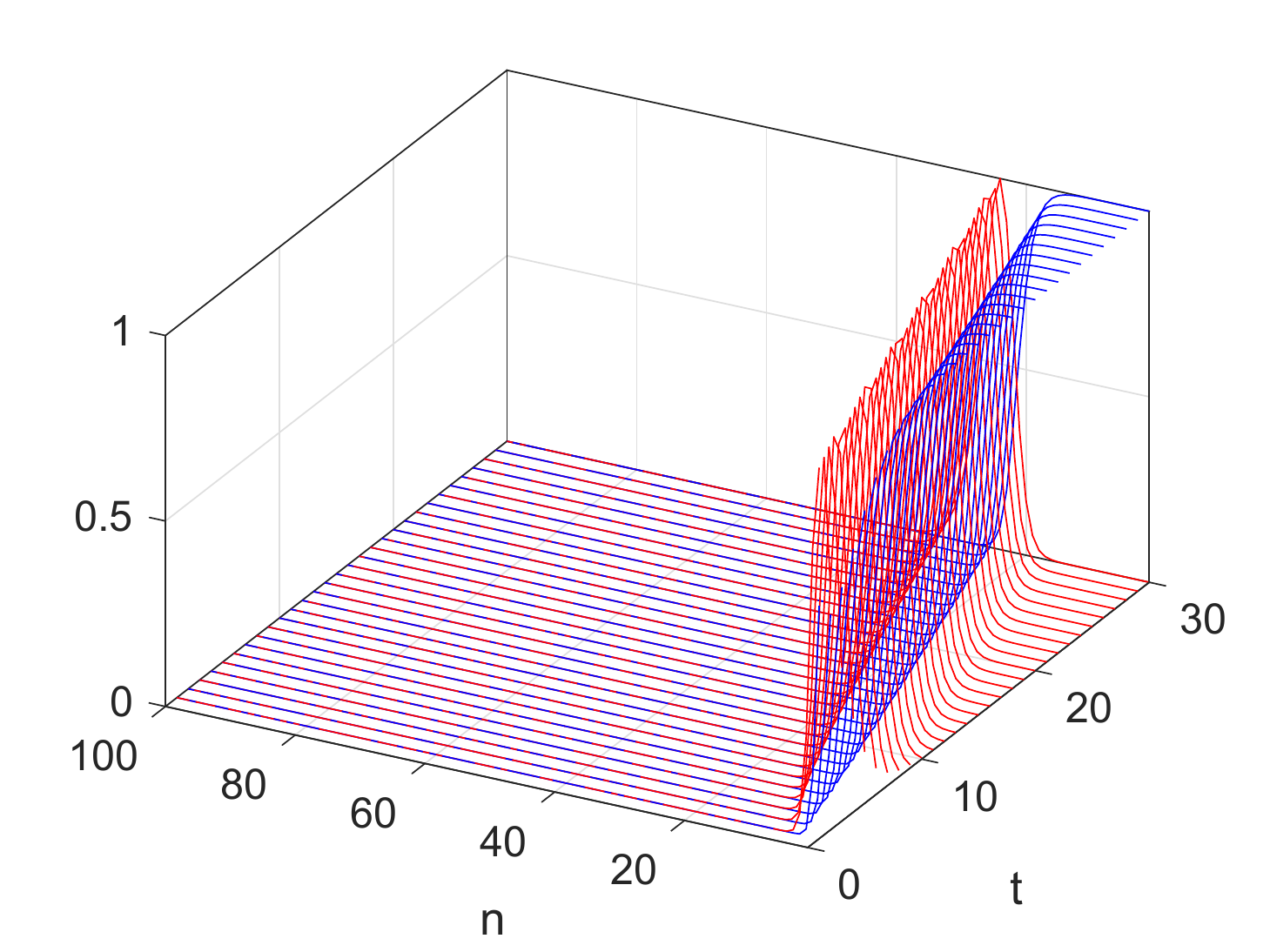}\hfill
\includegraphics[width=0.3\textwidth]{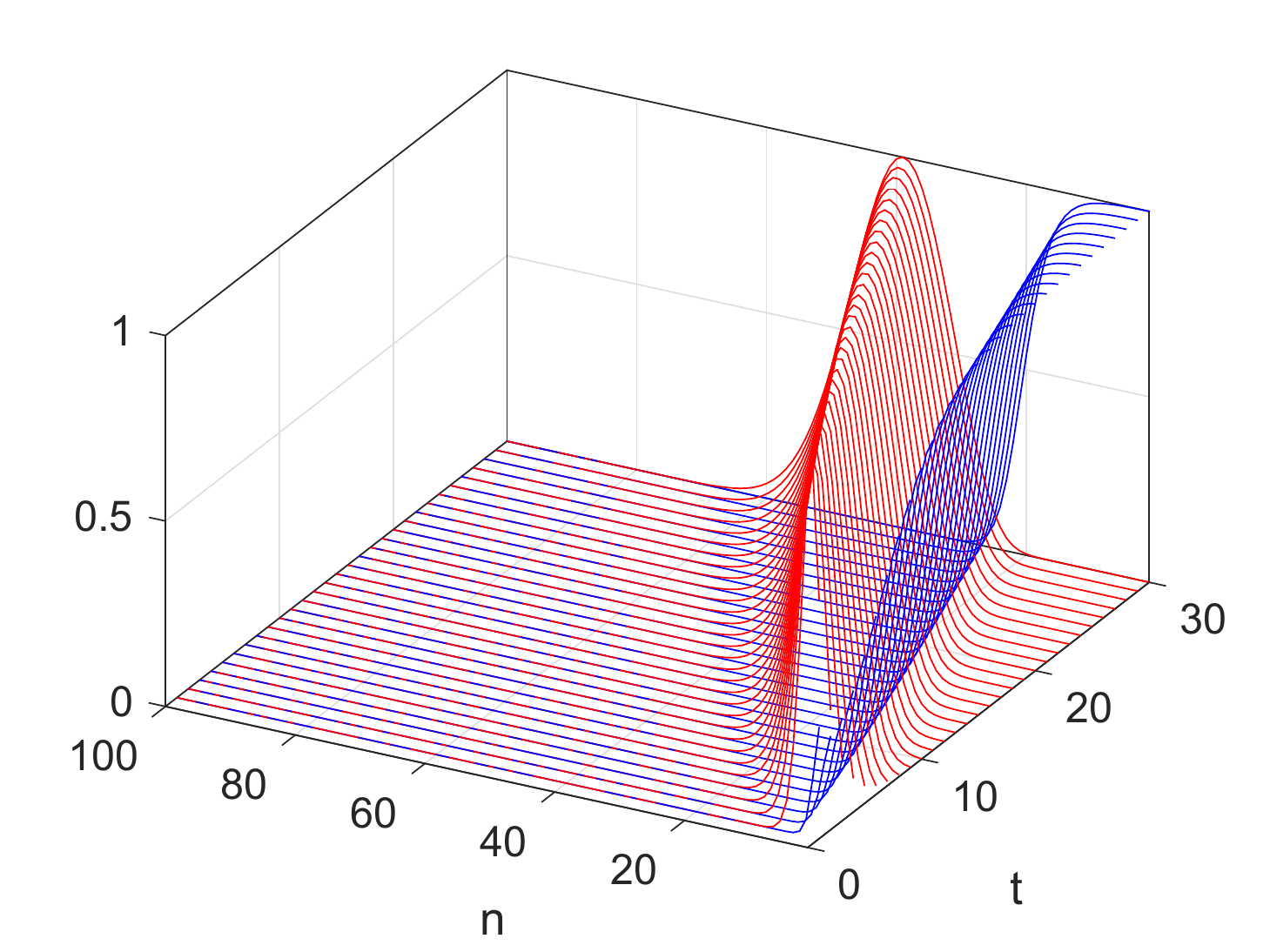}\hfill
\includegraphics[width=0.3\textwidth]{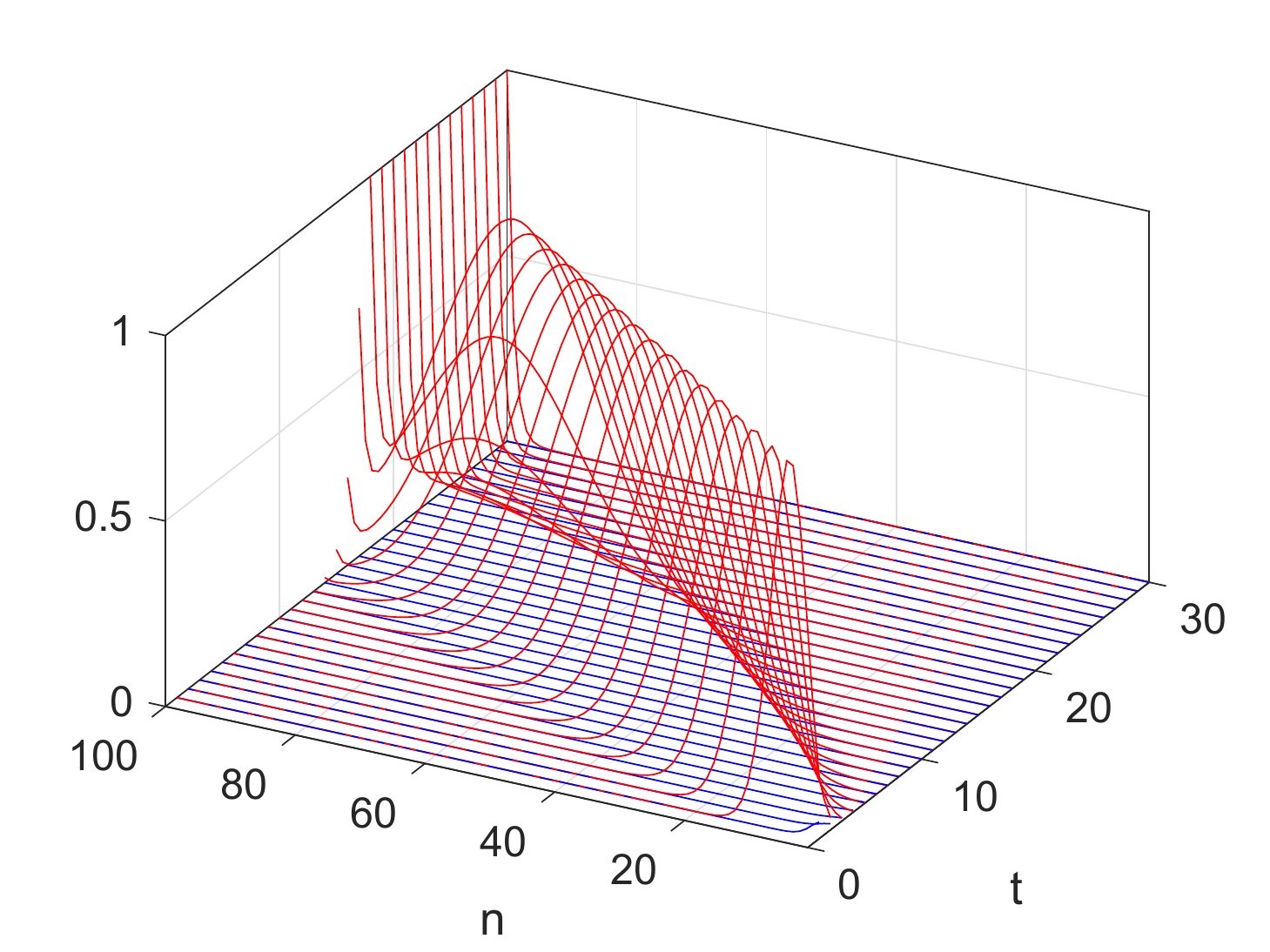}
\caption{Numerical simulations of (\ref{eq:tree}) with $k=3$ and for $\alpha=0.2$ (left), $\alpha=0.8$ (middle) and $\alpha=2.2$ (right).  The blue curves are $u_n(t)$ while the red curves depict the normalized population at each level, i.e. $w_n(t)/\mathrm{max}_n(w_n(t))$.  Note that $0.2<\alpha_1(3)<0.8<\alpha_2(3)<2.2$.  For $\alpha=0.2$, we observe that the maximal population is concentrated at the front interface. For $\alpha=0.8$, the maximal population is concentrated ahead of the front interface. Finally, for $\alpha=2.2$ the local population at any fixed node converges to zero, but the total population grows and eventually is concentrated at the final level of the tree.   }\label{fig:spacetime}
\end{figure}

\begin{figure}
   \includegraphics[width=0.4\textwidth]{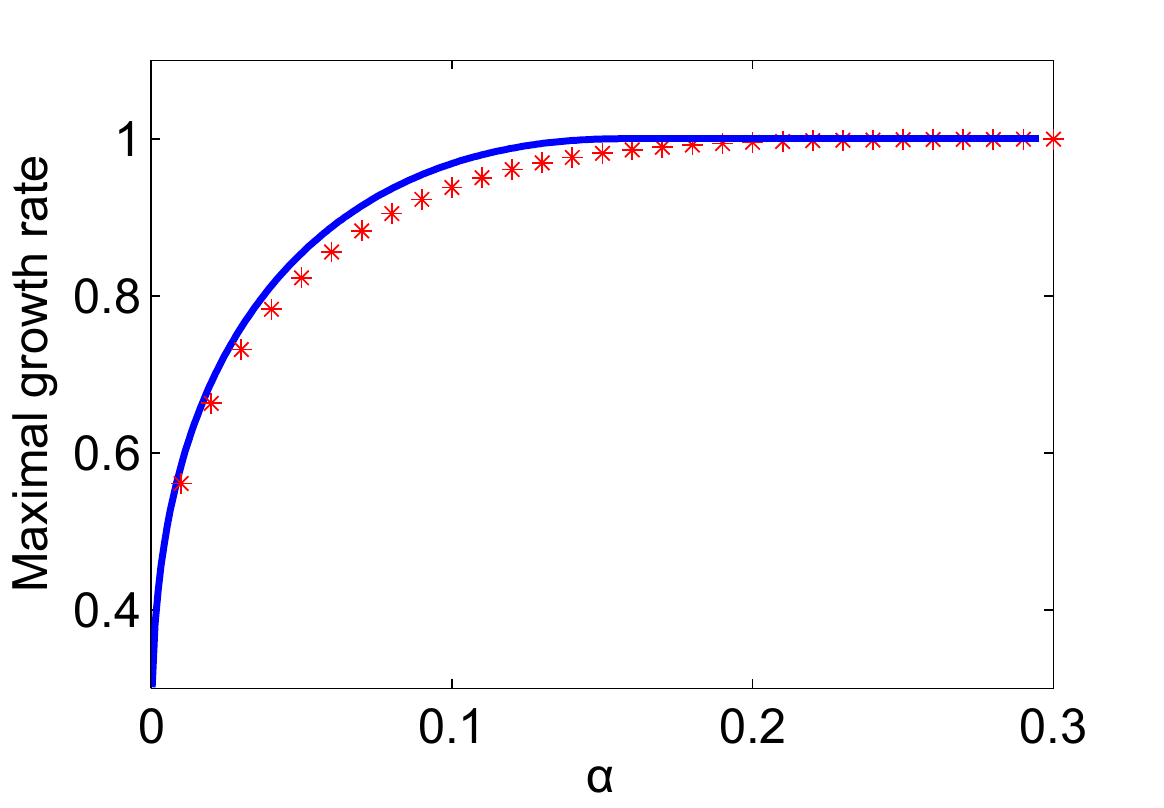} \hfill
   \includegraphics[width=0.45\textwidth]{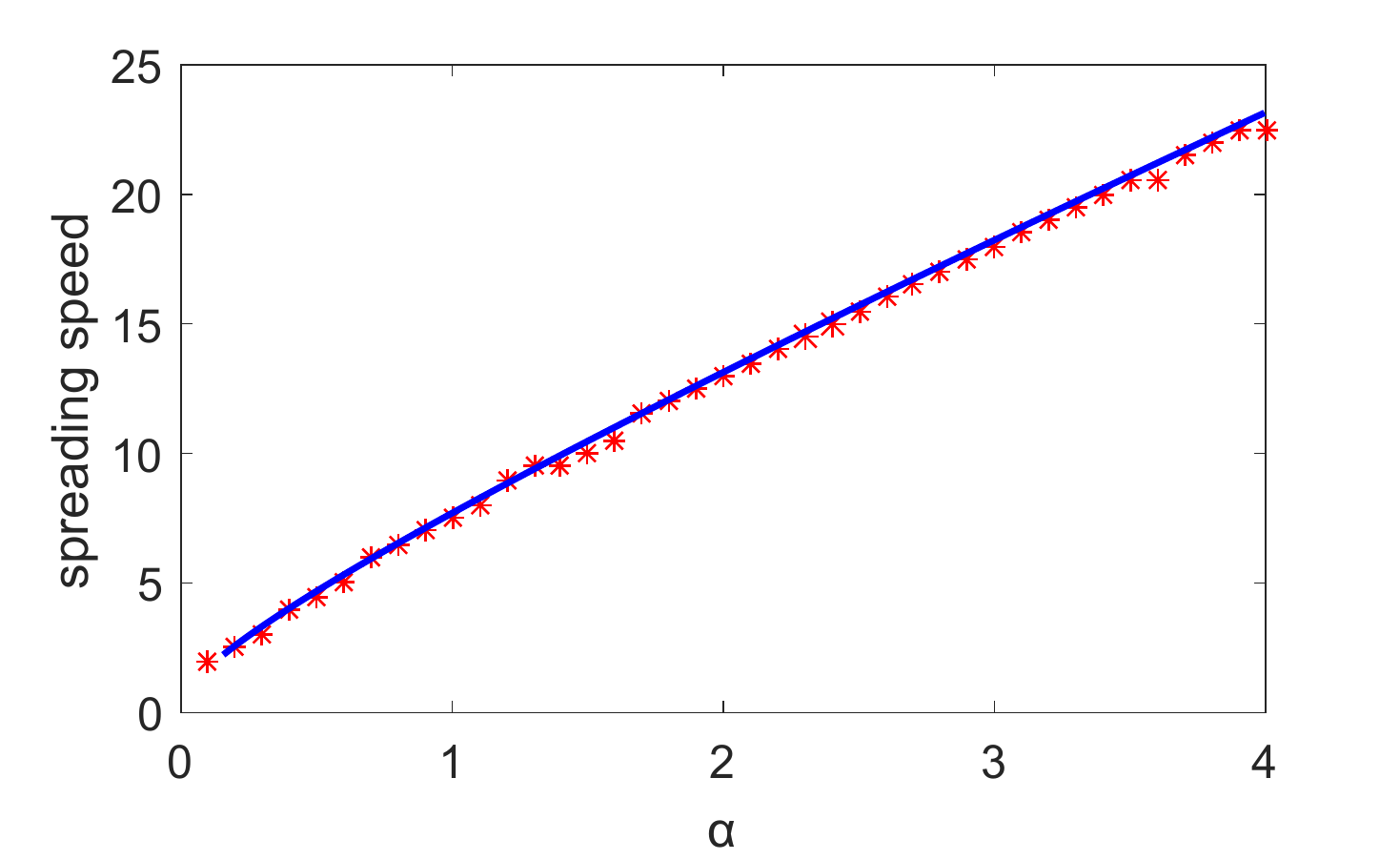}
\caption{On the left, we compare predictions for the exponential growth rate of the maximum of $w_n(t)$ as a function of $\alpha$ (blue line) against the exponential growth rates of $M(t)$ observed in direct numerical simulations (asterisks) for $k=5$.  On the right, we compare numerically observed spreading speeds for $w_n(t)$ (asterisks) versus linear spreading speeds determined numerically from the pinched double root criterion applied to $\tilde{d}_s(\gamma,\lambda)$ (blue line).  Here we have taken $k=5$. 
}
\label{fig:speedmass}
\end{figure}

\section{Erd\H{o}s-R\'enyi random graphs}\label{sec:ER}

We now turn our attention to more complicated network topologies.  We will use the analytical results obtained for the homogeneous tree to make predictions for the dynamics of the Fisher-KPP where $G$ is an Erd\H{o}s-R\'enyi random graph.  These predictions will then be compared to observations based upon numerical simulations.

The analysis performed on the homogeneous trees leads us to make two qualitative predictions for more general networks.  First, based upon the leading order independence of the linear spreading speed on the degree of the tree,  we expect that for $\alpha$ asymptotically  small the pointwise dynamics of the general system should be well approximated by a finite tree where the only relevant feature is the distance of a node from the location of initiation. Second, based upon our analysis in Section~\ref{sec:pop} we expect sub-linear growth rates for the total population for $\alpha$ small.  

We first review the Erd\H{o}s-R\'enyi random graph model \cite{erdos59}.   Let $N$ denote the total number of nodes in the graph and select some $p\in (0,1)$.  Erd\H{o}s-R\'enyi random graphs are constructed by assigning edges to the graph with fixed probability $p$.   We are interested in the case where $N\gg 1$ and the expected degree of each node $k_{ER}+1=Np$ is small.  We select one node at random and call that node the root.  We denote the solution at that node $u_1(t)$ and consider (\ref{eq:KPPonG}) with $u_1(0)=1$ and zero initial condition at every other node.

Several factors motivate our choice of Erd\H{o}s-R\'enyi random graphs for study.  Notable among these factors is their prevalence and popularity in the literature.  However, another important factor is the informal observation that Erd\H{o}s-R\'enyi random graphs with small $p$ appear locally tree-like: each node is connected  {\em roughly} $k_{ER}+1$ other nodes, which in turn will each be connected to {\em roughly} $k_{ER}$ which are themselves not connected with high probability.  We refer the reader to \cite{bollobas01,durrett07} and the reference therein for more rigorous studies of  Erd\H{o}s-R\'enyi graphs.

We construct networks with size ranging from $N=60,000$ to $N=500,000$, see \cite{batagelj}, and vary the expected degree from three ($k_{ER}=2$) to sixteen ($k_{ER}=15$).   Note that for the small  $p$ values considered here, the constructed networks are disconnected with probability one.  Therefore, when we discuss the network dynamics below we always restrict our attention to the largest connected component.  Numerically, solutions are computed using explicit Euler with timestep $\Delta t=.0025$.  Changing the timestep did not significantly alter the results.

\subsection{Pointwise dynamics: arrival times}
In this section, we focus on the pointwise dynamics of the system and describe numerical results characterizing the spreading speed on random graphs of Erd\H{o}s R\'eyni type.  

For finite graphs the only stable steady state is the uniform state at one.  We assign to each node a label $i\in \{1, 2, \dots, N\}$.  Let $d(i,j)$ be the shortest distance along the network between nodes $i$ and $j$.  Denote the node where the species is initialized as $i=1$ and specify initial conditions where $u_1(0)=1$ and $u_j(0)=0$ for all $j\neq 1$.  The primary descriptor of the dynamics in this situation is $\tau_i$, the {\em arrival time at node $i\in V$}, defined as
\[ \tau_i =\min_{t\geq 0} u_i(t)>\kappa,\]
for some threshold $0<\kappa<1$.  

Arrival times as a function of the distance from the initial node are plotted in Figure~\ref{fig:ATER}.  For very small $\alpha$, we observe a strong linear relationship between the arrival time and distance from the node $i=1$.   As $\alpha$ is increased this linear relationship becomes less apparent.  

We will compute the  {\em mean arrival time at level $l$},
\[ M_l=\frac{1}{N_l}\sum_{i\in V d(i,1)=l} \tau_i.\]
Based upon this data we find a best fit linear regression line as follows.  In order to minimize the effect of initial transients or boundary effects, we consider only those $M_l$
corresponding to distances greater than two from the initial node or less than two from the furthest node.  A linear regression is then computed and the reciprocal of the slope is called the numerically observed {\em mean spreading speed}.  These regression lines are shown in green in the Figure~\ref{fig:ATER} with green asterisks denoting $M_l$.

Mean spreading speeds are computed and plotted in Figure~\ref{fig:ERss}.  We compare these speeds to spreading speeds on homogeneous trees computed in Section~\ref{sec:tree}.  Before drawing any conclusions, some remarks regarding the relationship between random graphs and regular trees are in order.  Important in any comparison of spreading speeds in the random graph and the homogeneous tree is a selection of a value of $k$ to be used for comparison.  This choice is not obvious, nor is it clear that one such $k$ will be sufficient.  For large Erd\H{o}s-R\'enyi graphs, the average degree in the graph will converge to $k_{ER}+1$ as $N\to\infty$.  At the same time, however, if $k_{ER}$ is small then the Erd\H{o}s-R\'enyi graph will be disconnected with probability one.  Within the largest connected component the average degree will exceed $k_{ER}+1$.  The simplest choice is then to  take  $k$ to be the average degree of all nodes over which the linear regression is computed.  With this choice, the observed and theoretical speeds are close, however the linear spreading speed consistently overestimates the observed mean spreading speed.  This discrepancy could be due to several factors including the heterogeneity of the network. However, we also note that in the system (\ref{eq:tree}) it takes some time for initial data to converge to a traveling front and even then the convergence is notoriously slow \cite{bramson83}.  As a result, it is perhaps unrealistic to expect that spreading speeds measured directly in numerical simulations on the random graph to match exactly with predictions from the linear spreading speed.

We comment briefly on the dynamics for large values of $\alpha$.  For $\alpha$ sufficiently large the first arrival times cease to be monotonic with respect to distance from the original node.   For even larger values of $\alpha$ the arrival time for most nodes on the graph is approximately constant.  This is to be expected.  For $\alpha$ sufficiently large the dominant eigenvalue of the linearization about zero is one, whose eigenvector is the constant vector $\vec{1}$.  Projecting the initial condition onto this vector we expect leading order dynamics
\[ u(t)\sim \frac{1}{N}e^t \vec{1},\]
from which we find a homogeneous arrival time that scales with $\log N$.  This is consistent with the timescale observed for large values of $\alpha$.

\begin{figure}[ht]
\centering
 \subfigure[$\alpha=0.001$.]{\includegraphics[width=0.3\textwidth]{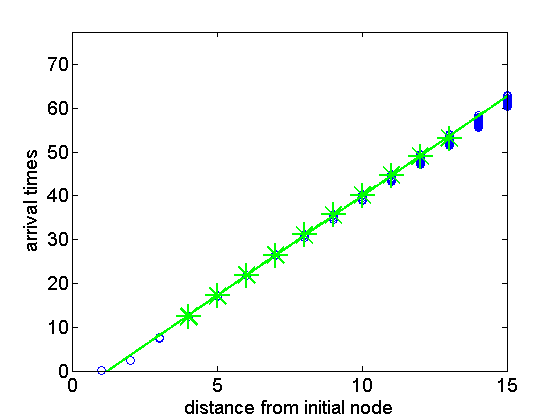}}
\subfigure[$\alpha=0.01$.]{\includegraphics[width=0.3\textwidth]{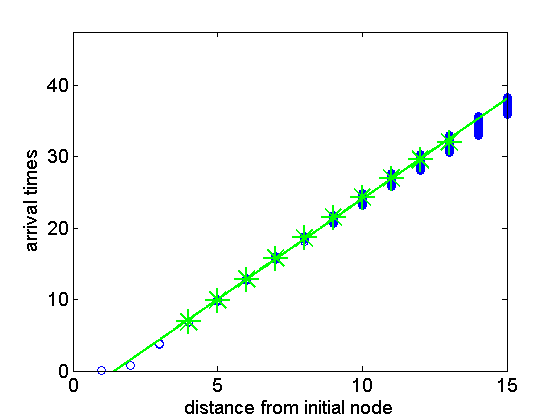}}
\subfigure[$\alpha=0.1$.]{\includegraphics[width=0.3\textwidth]{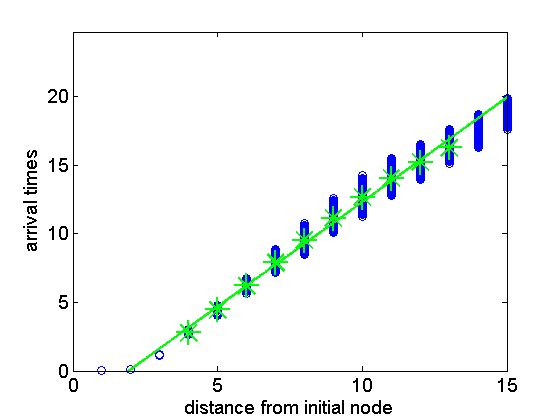}}
 \subfigure[$\alpha=0.2$.]{\includegraphics[width=0.3\textwidth]{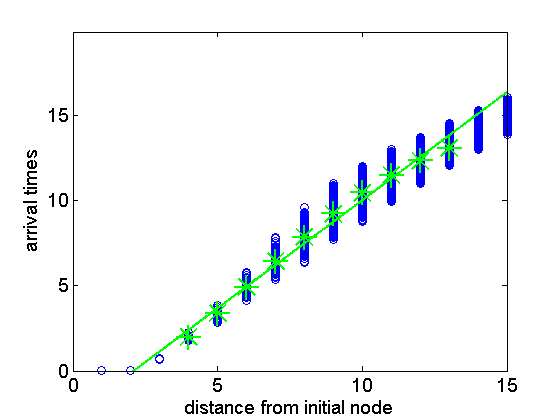}}
 \subfigure[$\alpha=0.5$.]{\includegraphics[width=0.3\textwidth]{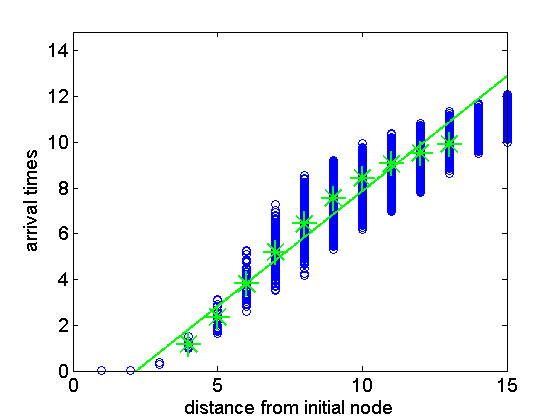}}
\subfigure[$\alpha=1.0$.]{\includegraphics[width=0.3\textwidth]{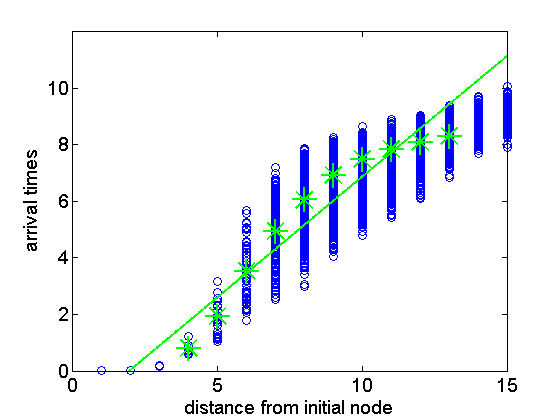}}
 \subfigure[$\alpha=2.0$.]{\includegraphics[width=0.3\textwidth]{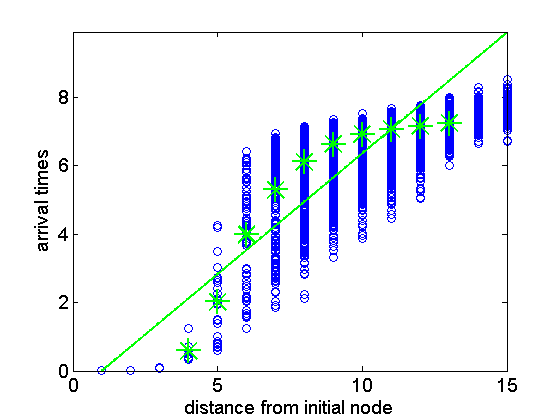}}
 \subfigure[$\alpha=4.0$.]{\includegraphics[width=0.3\textwidth]{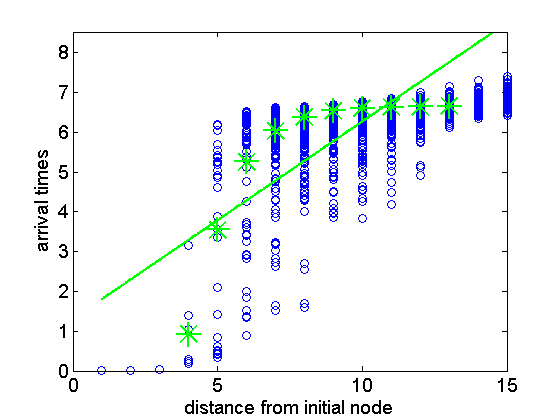}}
\caption{Arrival times for an Erd\H{o}s R\'eyni graph with $N=60,000$ and expected degree $k_{ER}=2$.  Various values of $\alpha$ are considered.  In green is the best fit linear approximation for the mean arrival times for nodes with distance between $3$ and $12$ from the initial location.      }
\label{fig:ATER}
\end{figure}

\begin{figure}
\includegraphics[width=0.4\textwidth]{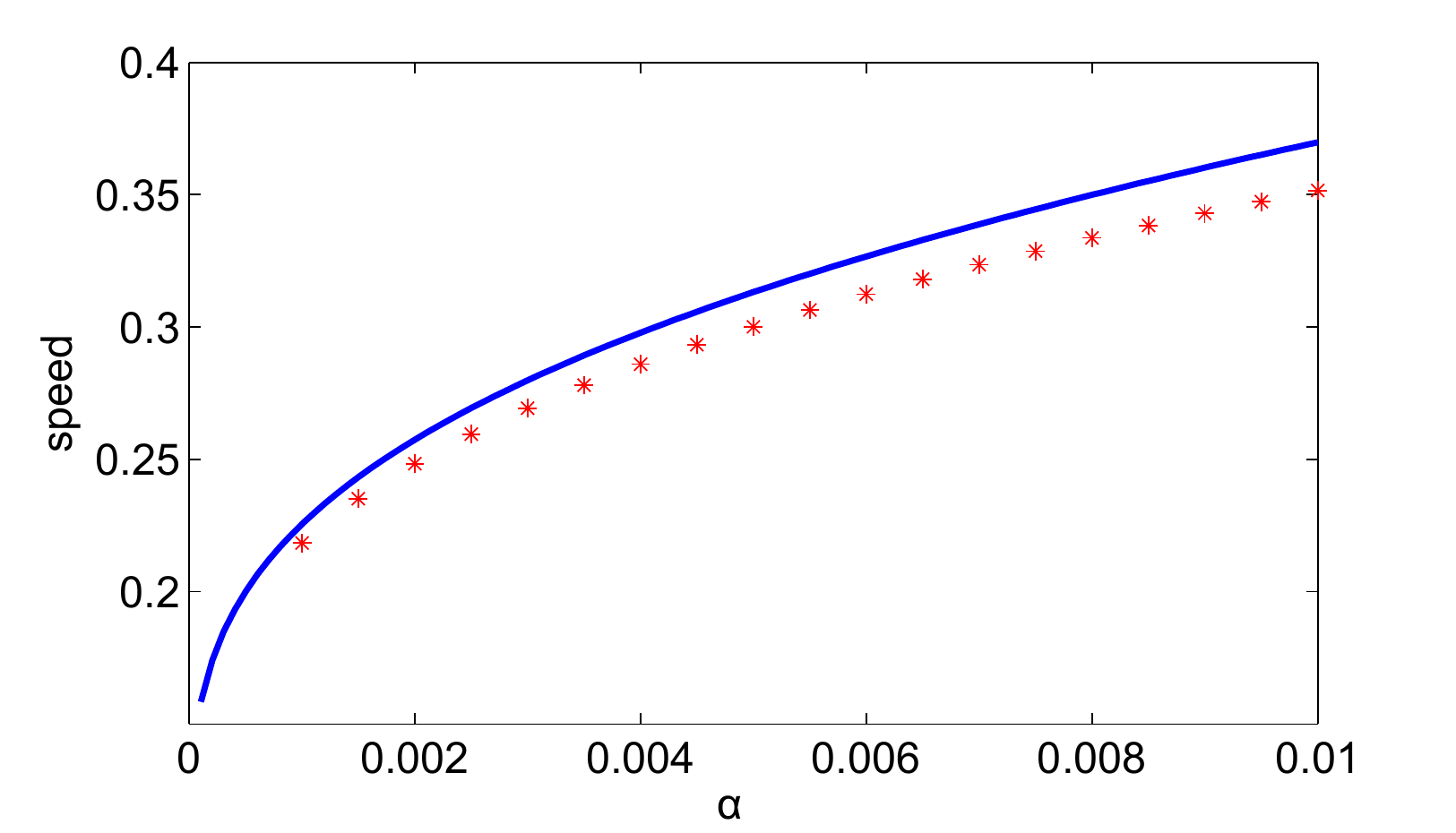}\hfill
\includegraphics[width=0.4\textwidth]{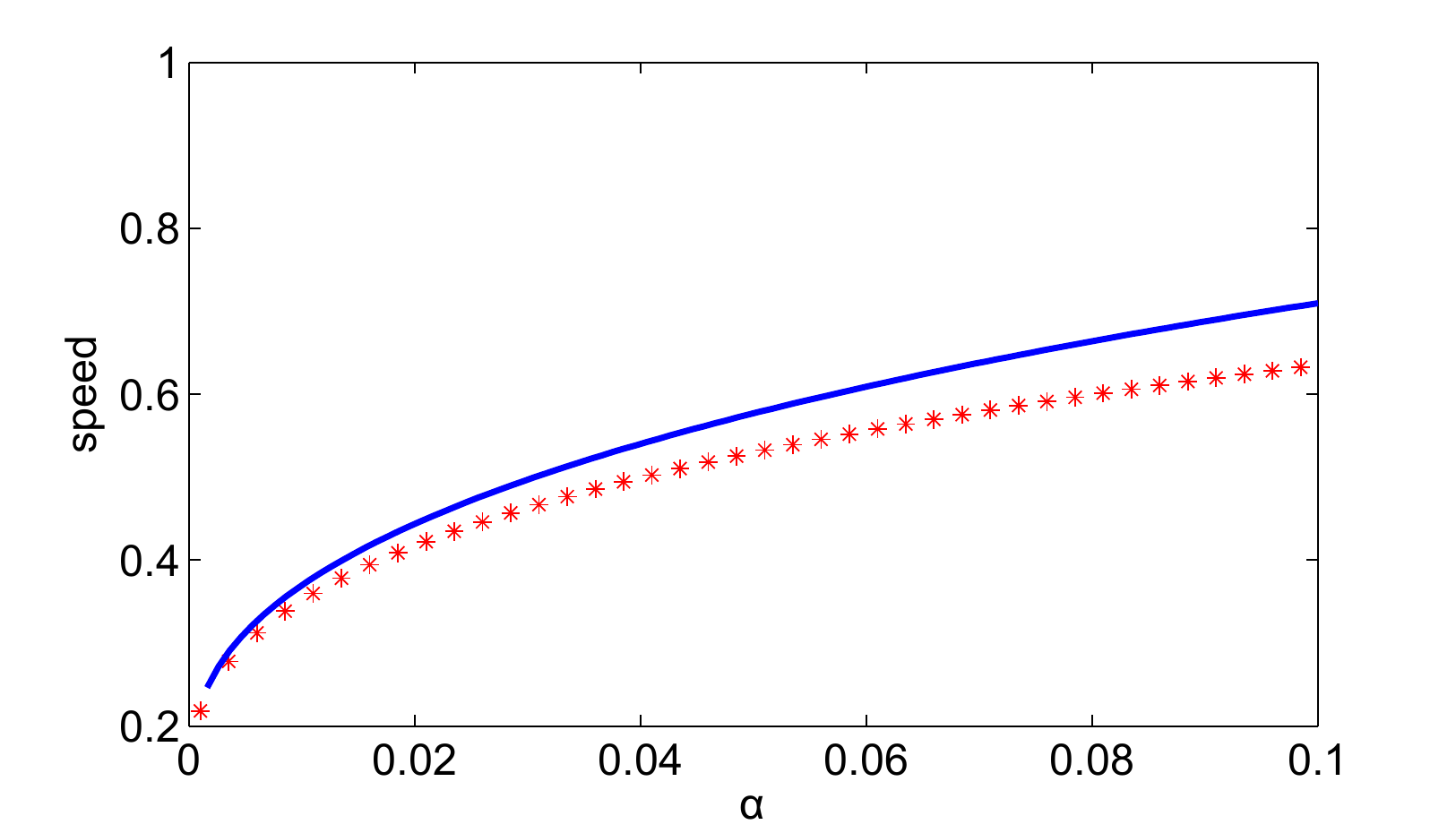}
\caption{Speed associated to the mean arrival times in numerical simulations on  an Erd\H{o}s R\'eyni graph with $N=60,000$ and expected degree $k_{ER}=2$ are shown in asterisks.  The blue curve is the spreading speed predicted by the analysis in Section~\ref{sec:tree}  for the homogeneous tree with $k=2.54$, found by numerically computing roots of (\ref{eq:F}).  This value is chosen since it is one less than  the mean degree of the network over those nodes with distance between $3$ and $12$ from the original location.      }\label{fig:ERss}
\end{figure}

\subsection{Aggregate dynamics: population growth rate}

We now investigate the exponential growth rate of the total population for (\ref{eq:KPPonG}) for Erd\H{o}s-R\'enyi random graphs.  We compare these growth rates to those observed on the homogeneous tree.  

The logarithm of the total population is shown in Figure~\ref{fig:loggrowth}.  After an initial transient, the growth rate appears to be roughly linear corresponding to exponential growth.  We measure this growth rate and compare it to similar computations in the homogeneous tree.  

As was the case for the pointwise analysis, an immediate challenge lies in determining the precise value of $k$ to use in the comparison to the homogeneous tree.   Given the variation in the average degrees across levels of the network, we instead compare numerically observed growth rates in the random graphs with observed growth rates in homogeneous trees that we expect to provide bounds on the total growth rate.  The results of these simulations are depicted in Figure~\ref{fig:ERgrowthrate}.

\begin{rmk} Population growth rates for the Fisher-KPP equation on Erd\H{o}s-R\'enyi random graphs were also studied in \cite{burioni12}.  Recall that sub-linear growth rates were found for the homogeneous tree; see Section~\ref{sec:GRnonlinear} where  sub-linear growth rates are explained by the existence of finite mass traveling fronts.  We suggest the same mechanism is at play for Erd\H{o}s-R\'enyi random graphs although we do not pursue making this mathematically rigorous.  We also note that the analysis on the tree suggests a growth rate that scales with $s_{lin}(\alpha)\log(k)$ for small $\alpha$, which for small $\alpha$ has leading order asymptotic expansion of $-\log(\alpha)\log(k)$.  This should be contrasted with the polynomial scaling of the exponential growth rate with respect to $\alpha$  suggested in \cite{burioni12}.  
\end{rmk}

\begin{figure}
\includegraphics[width=0.45\textwidth]{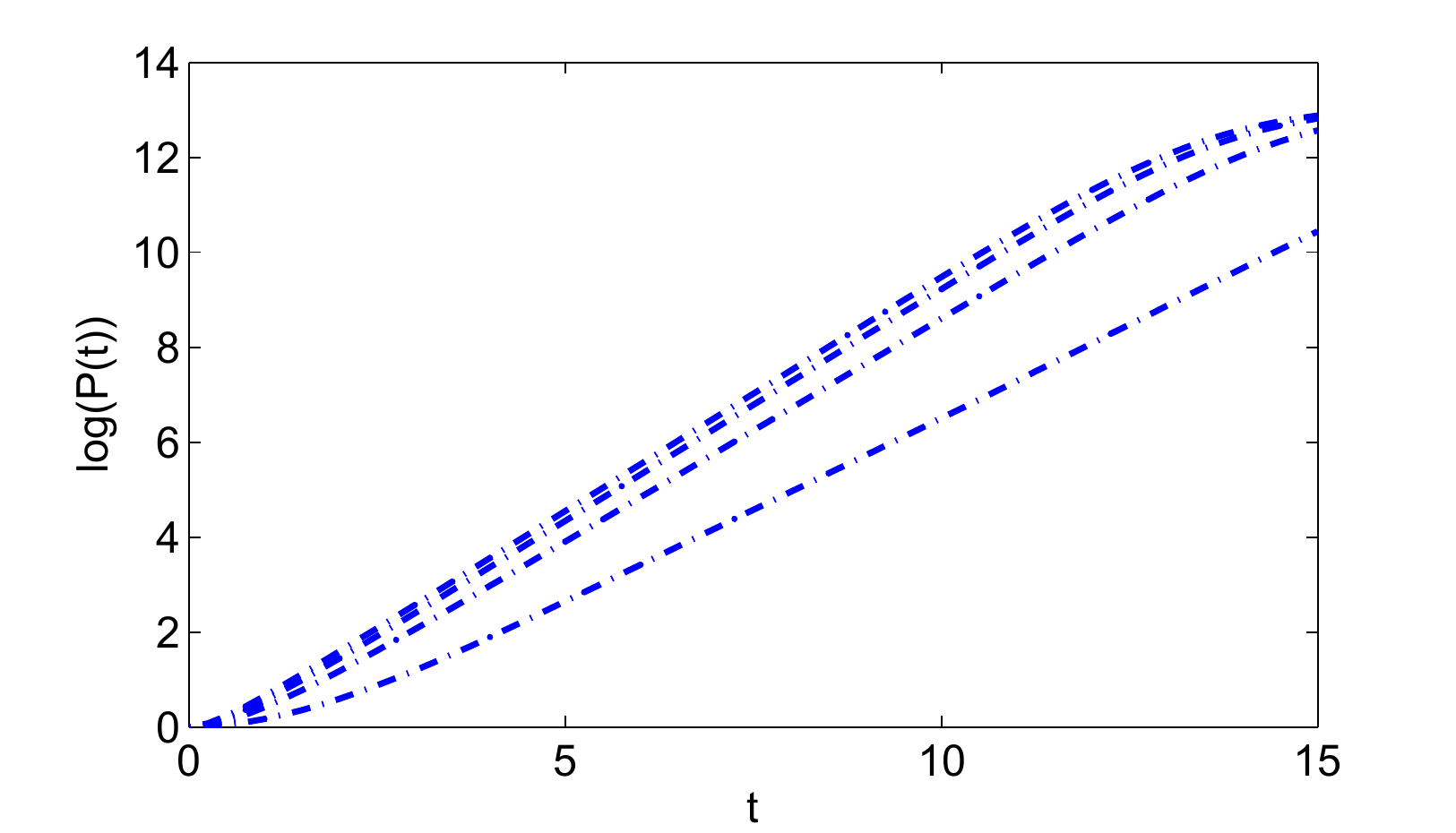}\hfill
\includegraphics[width=0.45\textwidth]{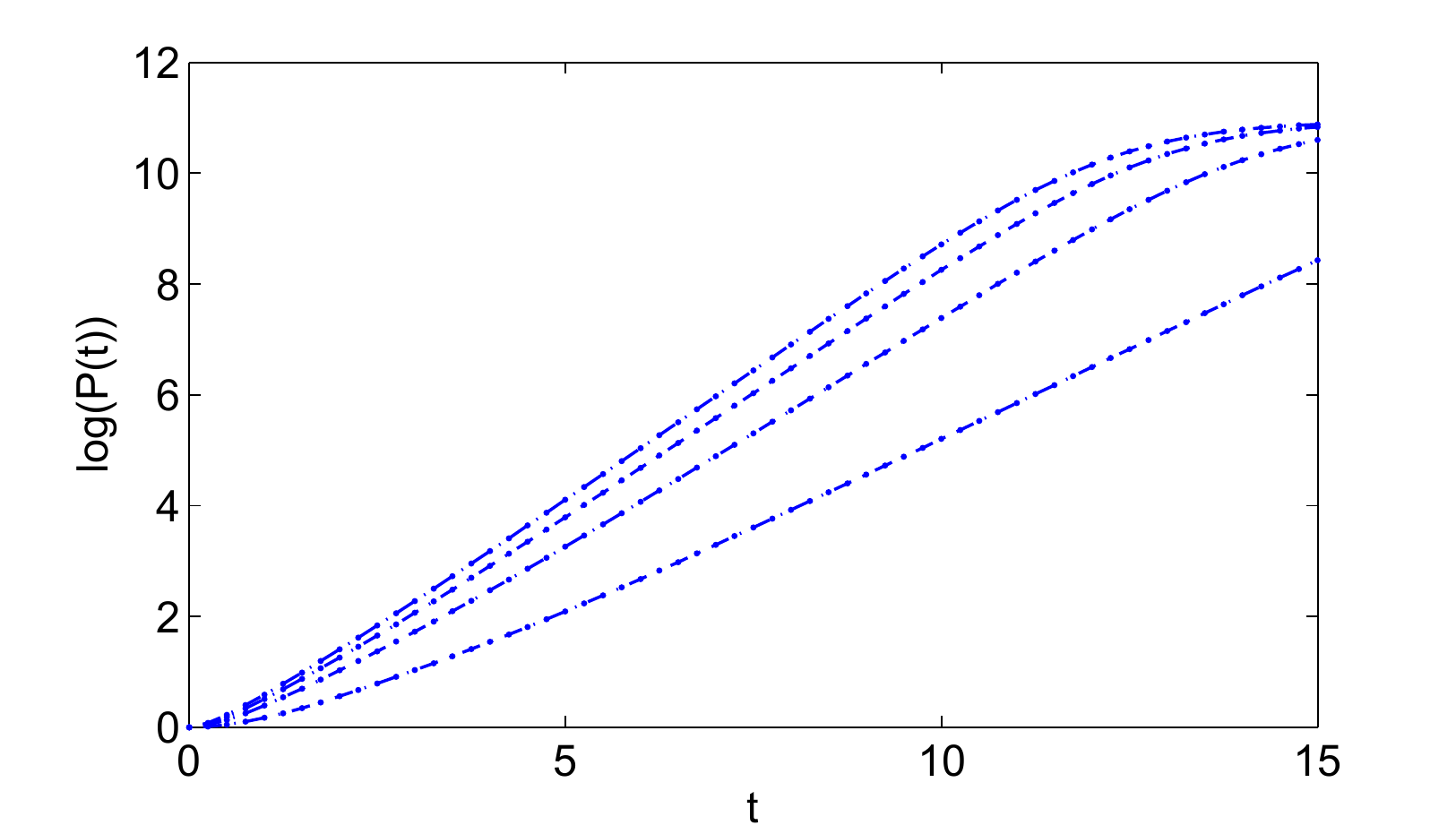}
\caption{Growth rate of the total population for Erd\H{o}s-R\'enyi graph.  On the left, $N=500,000$ and $\alpha=0.1,0.35,0.6,0.85$.  Larger values of $\alpha$ correspond to faster growth rates.  On the right is the case of $N=60,000$ with the same values of $\alpha$.   }\label{fig:loggrowth}
\end{figure}

\begin{figure}
\includegraphics[width=0.45\textwidth]{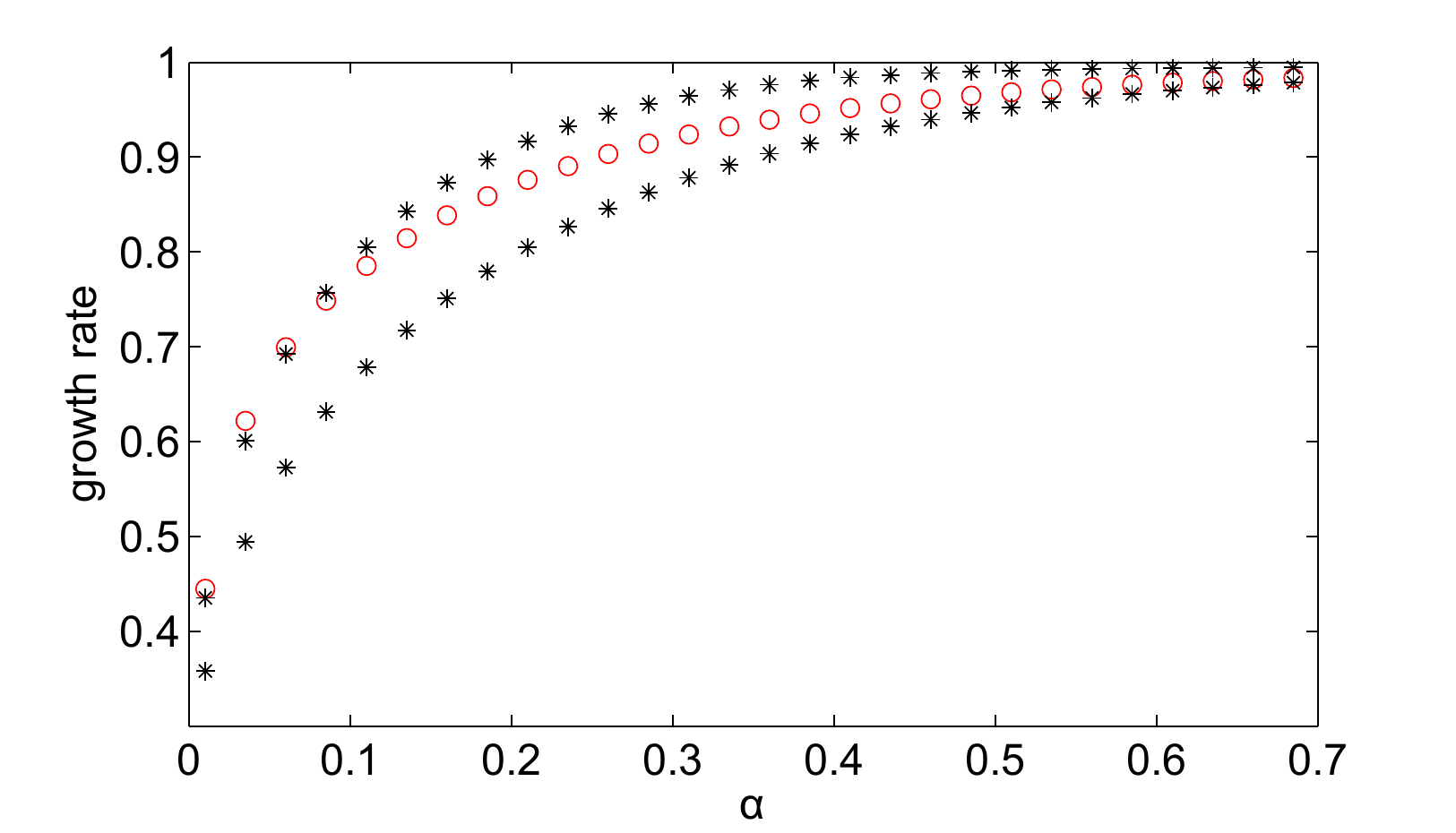}\hfill
\includegraphics[width=0.45\textwidth]{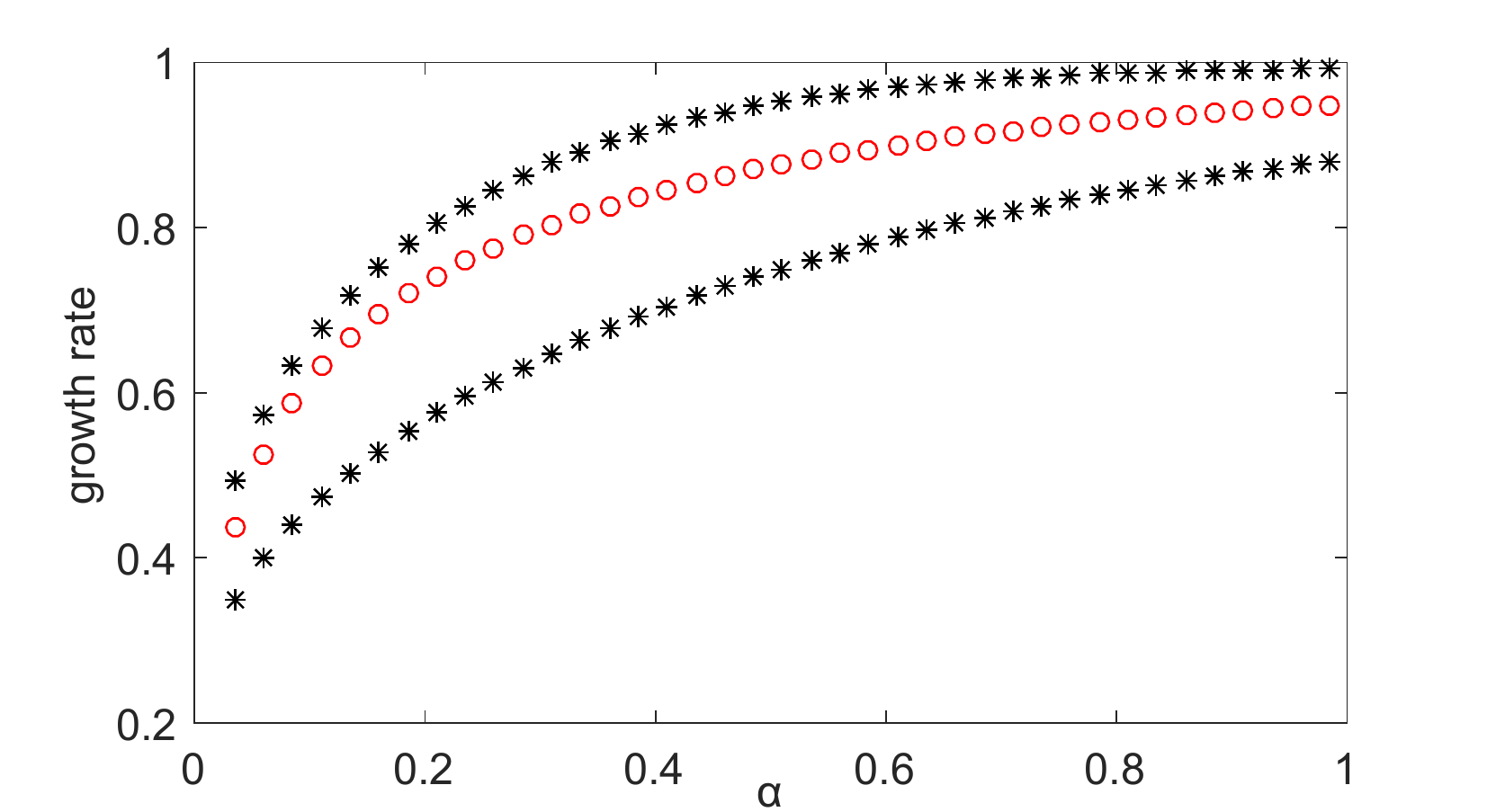}
\caption{Numerically calculated exponential growth rate for the Erd\H{o}s-R\'enyi graph.  On the left, $N=500,000$ and observed growth rates are plotted as circles.  The asterisks are the corresponding growth rates in the homogeneous tree with depth $13$.  The lower curve is degree $k=3$ while the larger curve is degree $k=4$.  On the right are the same computations, but for the Erd\H{o}s-R\'enyi graph with $N=60,000$ and for homogeneous trees with $k=2$ and $k=3$.   }\label{fig:ERgrowthrate}
\end{figure}

\section{Discussion and outlook for future work}\label{sec:conc}

In this paper, we studied the dynamics of the  Fisher-KPP equation defined on homogeneous trees and Erd\H{o}s-R\'{e}yni random graphs.  For the homogeneous tree we study traveling fronts and spreading speeds.  We find that these fronts are linearly determined and propagate at the linear spreading speed.  An interesting property is that these speeds are non-monotone with respect to the diffusion parameter and eventually become zero.  Therefore, for values of $\alpha>\alpha_2$ the solution converges pointwise exponentially fast to zero.  We find that the critical point of the linear spreading speed at $\alpha=\alpha_1$ occurs for the parameter value where the maximal linear growth rate of the total population in the network occurs at the linear spreading speed.  For larger values of $\alpha$ the front is slowed due to some fraction of the population escaping down the graph.  The increased mobility of the species allows the population to explore nodes with only small concentrations of the species where the local growth rate is maximal.

An interesting feature of the linear spreading speed is that it is independent of the degree of the nodes in the network in the asymptotic limit as $\alpha\to 0$.  This suggested that the spreading properties observed in the homogeneous tree should be present in random graphs and explored this for the case of random networks of Erd\H{o}s-R\'{e}yni type.  For $\alpha$ asymptotically small, we observed traveling fronts and characterized their speed.  It was also the case that the exponential growth rate of the population was less than the maximal growth rate of one and bounded this growth rate by the corresponding growth rate on the homogeneous tree.

Erd\H{o}s-R\'{e}yni graphs are known to share a close relationship to trees.  It would be interesting to extend the results here to different, more realistic, random networks such as small world \cite{watts98} or scale free networks \cite{barabasi99}.  The independence of the spreading speed on the degree for small $\alpha$ suggests that one should continue to observe traveling fronts on these more general networks that are well approximated by the lattice Fisher-KPP equation.  The situation is more complicated for larger values of $\alpha$ where the heterogeneity of the network can not be denied.  It would be interesting to characterize whether critical diffusion rates  analogous to $\alpha_1$ and $\alpha_2$ for these more general classes of graphs.  Example~\ref{ex:per} gives some indication that the effect of heterogeneity can be quite severe and so we expect that analysis of these systems would require a more general approach.

\section*{Acknowledgements} Portions of this research were initiated as part of a summer undergraduate research program at George Mason University as part of the EXTREEMS-QED program (NSF-DMS-1407087).  MH is grateful to Robert Truong for insightful discussions related to this project.  MH received partial support from the National Science Foundation through grant NSF-DMS-1516155.

\begin{appendix}
\section{Proof of Theorem~\ref{thm:spread}}\label{sec:proof}
In this section, we prove Theorem~\ref{thm:spread}.  The proof relies on the existence of a comparison principle for (\ref{eq:tree}).

\begin{prop} \label{prop:comp}
Suppose that there exists functions $\bar{u}_n(t)$ and $\underbar{u}_n(t)$ such that $0\leq \underbar{u}_n(0)\leq \bar{u}_n(0)\leq 1$.  Furthermore suppose that for all $t\geq 0$ and all  $n\geq 2$  we have 
\begin{eqnarray*} \bar{u}_n'&\geq &\alpha\left(\bar{u}_{n-1}-(k+1)\bar{u}_n+k\bar{u}_{n+1}\right)-f(\bar{u}_n)  \\
\underbar{u}_n'&\leq &\alpha\left(\underbar{u}_{n-1}-(k+1)\underbar{u}_n+k\underbar{u}_{n+1}\right)- f(\underbar{u}_n),
\end{eqnarray*}
while for $n=1$,
\begin{eqnarray*} \bar{u}_1'&\geq &\alpha\left(-k\bar{u}_1+k\bar{u}_{2}\right)- f(\bar{u}_1)  \\
\underbar{u}_1'&\leq & \alpha\left(-k\underbar{u}_1+k\underbar{u}_{2}\right)- f(\underbar{u}_1).
\end{eqnarray*}
Then for all $t\geq 0$ we have that 
\[ 0\leq \underbar{u}_n(t)\leq \bar{u}_n(t)\leq 1. \]

\end{prop}
We refer the reader to \cite{chen97} for a proof of the comparison principle in a setting similar to the one here.  

We now proceed to the proof of Theorem~\ref{thm:spread}.  We find it convenient to break the analysis according to whether $\alpha<\alpha_2$ or $\alpha>\alpha_2$.

\paragraph{Case I: $\alpha<\alpha_2(k)$}

\begin{lem}\label{lem:sup} Let $\gamma>0$ such that $s_{env}(\gamma)>0$.  Consider any $\alpha<\alpha_2$.  Then 
\[ \bar{u}_n(t)=\min\{1, Ce^{-\gamma(n-s_{env}(\gamma)t)}\},\]
is a super-solution for all $n\in\mathbb{N}$, any $C>1$ and all $t\geq 0$. 
\end{lem}

\begin{proof} We need to show that for $n\geq 2$
\[ N(\bar{u}_n)=u_n'-\alpha\left(u_{n-1}-(k+1)u_n+ku_{n+1}\right)- u_n-\left(f(u_n)-u_n\right)\geq 0.\]
Since the constant $1$ is a solution of (\ref{eq:tree}), it is by definition a super-solution and $N(1)=0$.  The exponential term propagates at the envelope velocity of the mode $\gamma$.  We recall that 
\[ \alpha\left(e^\gamma-k-1+ke^{-\gamma}\right)-\gamma s_{env}(\gamma) +1=0.\]
As a result, we have that the linear terms in $N(\bar{u}_n)$ are zero and 
\[ N\left(Ce^{-\gamma(n-s_{env}(\gamma)t)}\right)=-\left(f(u_n)-u_n\right), \]
which is positive by the KPP assumption  $f(u)<f'(0)u$.  Thus, the exponential is a super-solution for all $n\geq 2$.  Since $C>1$, the super-solution is always one at the root and we have that $\bar{u}$ is a super-solution.

\end{proof}

We now turn our attention to the construction of sub-solutions.  Let $0<\mu<1$ and consider the linear equation
\[ \frac{d\phi_n}{dt}=\alpha\left(\phi_{n-1}-(k+1)\phi_n+k\phi_{n+1}\right)+\mu\phi_n.\]
Exponential solutions $e^{-\gamma(n-st)}$ are obtained for $\gamma$ and $s$ satisfying the following equation 
\be \alpha\left(e^{\gamma}-(k+1)+ke^{-\gamma}\right)-s\gamma+\mu=0.\label{eq:gsub}\ee
For each $\mu$, there exists $s_\mu$ such that (\ref{eq:gsub}) has no real solutions for $s<s_\mu$.  Let $\gamma_\mu\in\mathbb{C}$ be a complex solution.  Then let
\be \phi_\mu (y)=e^{\gamma_{\mu}y}+c.c. ,\label{eq:phi}\ee
where $c.c.$ denotes the complex conjugate and fix $a<b$ such that $\phi(y)>0$ for $y\in(a,b)$ and $\phi(a)=\phi(b)=0$. We use these functions to construct sub-solutions in the following Lemma.

\begin{lem}\label{lem:sub} Fix $\alpha<\alpha_2(k)$ and let $s<s_{lin}$.  Let $\mu<1$ such that $s<s_{\mu}$.  Then there exists  $T>0$,  a function $\phi_\mu(y)$ and an $\e^*$ such that 
\be \underbar{u}_n(t)=\left\{ \begin{array}{cc} \e \phi_\mu(n-st) & a\leq n-st\leq b \\
   0 & \text{else} \end{array}\right. \label{eq:usub}\ee
is a sub-solution for all $0<\e<\e^*$ and for all $t>T$.
\end{lem}
\begin{proof}
For $s<s_{lin}$, select $\mu<1$ so that $s<s_\mu$ and consider the function $\phi_\mu$ in (\ref{eq:phi}) whose support is the interval $[a+st,b+st]$.  Let $T$ be sufficiently large so that $sT+a>1$.  Now consider
\[ N(u_n)=u_n'-\alpha\left(u_{n-1}-(k+1)u_n+ku_{n+1}\right)- u_n-\left(f(u_n)-u_n\right).\]
We find
\[ N(\underbar{u}_n)=-\e \phi(n-st)(1-\mu)+\left(\e \phi(n-st)-f(\e\phi(n-st))\right).\]
For $u$ small, there exists a $C>0$ such that $-Cu^2<f(u)-u$.  Let $\e$ be sufficiently small so that this bound holds for (\ref{eq:usub}).  Then
\[ N(\underbar{u}_n)<\e \phi(n-st)\left(-(1-\mu)+C\e \phi(n-st)\right),\]
and if we restrict
\[ \e<\frac{1-\mu}{C\max_{y\in[a,b]} \phi(y)},\]
then $N(\underbar{u}_n)<0$ and $\underbar{u}_n$ is a sub-solution.  
\end{proof}

\paragraph{Case II: $\alpha>\alpha_2(k)$} 

For $\alpha>\alpha_2$, we only concern ourselves with the establishment of super-solutions.  

\begin{lem}\label{lem:sup2} Consider any $\alpha>\alpha_2$.  Let $\gamma>0$ such that $s_{env}(\gamma)<0$.   Then 
\[ \bar{u}_n(t)= e^{-\gamma(n-s_{env}(\gamma)t)},\]
is a super-solution for all $n\in\mathbb{N}$ and all $t\geq 0$. 
\end{lem}

\begin{proof}
We follow the proof of Lemma~\ref{lem:sup2} and compute $N(u_n)$.  The calculation is exactly the same, aside from the root.  There we calculate
\begin{eqnarray*} N(u_1)&=& \left(s_{env}(\gamma)\gamma - \alpha(k+1)(-1+e^{-\gamma})-1\right)e^{-\gamma(1-s_{env}(\gamma)t)} -\left(f(u_n)-u_n\right) \\
&=& \alpha e^{-\gamma(1-s_{env}(\gamma)t)} \left(e^\gamma-e^{-\gamma}\right)-\left(f(u_n)-u_n\right).
\end{eqnarray*}
Since both terms are positive, we have that $N(u_1)$ is positive as well and the proof is completed.

\end{proof}

We now prove Theorem~\ref{thm:spread}.  The claim of pointwise convergence for $\alpha>\alpha_2(k)$ follows from Lemma~\ref{lem:sup2} and the fact that $s_{env}(\gamma)<0$.  Consider then $\alpha<\alpha_2$.  Here, we have that the linear spreading speed is positive and obtained for some value of $s_{lin}$ with selected decay rate $\gamma_{lin}$.  Then since $C>1$ in Lemma~\ref{lem:sup} we have that $\bar{u}_n(0)\geq u_n(0)$.  Due to the comparison principle, this relationship holds for all time and we obtain an upper bound on the location of the invasion point: $\kappa(t)\leq s_{lin}(\alpha)t $.  To obtain a lower bound, we consider the sub-solutions constructed in Lemma~\ref{lem:sub}.   Select $s<s_{lin}(\alpha)$ and consider $\underbar{u}_n$ from Lemma~\ref{lem:sub}.  Let $t_1>T$.  By the maximum principle, we have that $u_n(t_1)>0$ for all $n\in\mathbb{N}$. Since $\underbar{u}_n(t_1)$ is compactly supported, we can select $\epsilon>0$ sufficiently small such that $\underbar{u}_n(t_1)\leq u_n(t_1)$.  Therefore, we have that $\kappa(t)\geq st$.  This holds for any $s<s_{lin}(\alpha)$ and we find $s_{lin}(\alpha)\leq s_{sel}\leq s_{lin}(\alpha)$ and Theorem~\ref{thm:spread} is established.  

\end{appendix}

\bibliographystyle{abbrv}
\bibliography{Treebib}

\end{document}